\documentclass[11pt]{article}

\usepackage{amsmath}
\usepackage{amsthm}
\usepackage{amssymb}
\usepackage{xcolor}
\usepackage{algorithm}
\usepackage{algorithmic}
\usepackage{comment}

\theoremstyle{plain}
\newtheorem{theorem}{Theorem}
\newtheorem{conjecture}[theorem]{Conjecture}
\newtheorem{lemma}[theorem]{Lemma}
\newtheorem{proposition}[theorem]{Proposition}
\newtheorem{corollary}[theorem]{Corollary}

\theoremstyle{definition}
\newtheorem{definition}[theorem]{Definition}
\newtheorem{example}[theorem]{Example}

\theoremstyle{remark}
\newtheorem{remark}[theorem]{Remark}

\newcommand{\tr}{\operatorname{tr}}
\newcommand{\uinvnorm}{|\kern-1pt|\kern-1pt|}

\newcommand{\poly}{\operatorname{poly}}

\newcommand{\supp}{\operatorname{supp}}

\newcommand{\sgn}{\operatorname{sgn}}
\newcommand{\wgt}{\operatorname{wgt}}

\newcommand{\R}{\mathbb{R}}

\newcommand{\C}{\mathbb{C}}

\newcommand{\E}{\mathbb{E}}
\newcommand{\fhats}{\hat{f}_\mathbf{s}}
\newcommand{\fhatt}{\hat{f}_\mathbf{t}}
\newcommand{\ghats}{\hat{g}_\mathbf{s}}

\newcommand{\chis}{\chi_\mathbf{s}}
\newcommand{\chit}{\chi_\mathbf{t}}

\def\1{\mathbb{I}}

\newcommand{\ket}[1]{| #1 \rangle}

\newcommand{\ip}[2]{\langle #1, #2 \rangle}
\newcommand{\vip}[2]{\langle #1|#2 \rangle}

\renewcommand{\Im}{\operatorname{Im}}

\newcommand{\be}{\begin{equation}}
\newcommand{\ee}{\end{equation}}
\newcommand{\bea}{\begin{eqnarray}}
\newcommand{\eea}{\end{eqnarray}}
\newcommand{\bes}{\begin{equation*}}
\newcommand{\ees}{\end{equation*}}
\newcommand{\beas}{\begin{eqnarray*}}
\newcommand{\eeas}{\end{eqnarray*}}

\begin{document}

\bibliographystyle{amsalpha}

\date{\today}

\title{Quantum boolean functions}

\author{Ashley Montanaro\footnote{Department of Computer Science, University of Bristol, Woodland Road, Bristol, BS8 1UB, UK; {\tt montanar@cs.bris.ac.uk}.} \ and Tobias J.\ Osborne\footnote{Department of Mathematics, Royal Holloway, University of London, Egham, TW20 0EX, UK; {\tt tobias.osborne@rhul.ac.uk}.}}

\maketitle

\begin{abstract}
In this paper we introduce the study of \emph{quantum boolean
functions}, which are unitary operators $f$ whose square is the
identity: $f^2 = \1$. We describe several generalisations of
well-known results in the theory of boolean functions, including
quantum property testing; a quantum version of the Goldreich-Levin
algorithm for finding the large Fourier coefficients of boolean
functions; and two quantum versions of a theorem of Friedgut, Kalai
and Naor on the Fourier spectra of boolean functions. In order to
obtain one of these generalisations, we prove a quantum extension of
the hypercontractive inequality of Bonami, Gross and Beckner.
\end{abstract}

\tableofcontents

%------------------------------------------------------------------------------x
%------------------------------------------------------------------------------

\section{Introduction}

Boolean functions stand at the crossroads between
many areas of study, such as social science, combinatorics,
computational complexity, and statistical mechanics, to name but a
few (see, e.g., \cite{odonnellnotes} or \cite{dewolf:2008a} for an introduction). The theory
of boolean functions has reached a maturity of sorts, and the
foundational results of the subject are now well established.

The advent of quantum algorithms (see, e.g., \cite{nielsen:2000a})
has added a twist to the computational complexity theory landscape
of which boolean functions are an integral part, and there are many
reasons to expect that this development will have more than
superficial consequences. Indeed, quantum algorithmic techniques
have already had important consequences within the field of
classical complexity theory alone, leading to a simple proof of the
closure of {\sf PP} under intersection \cite{aaronson:2005b} as well
as proofs of new lower bounds for locally decodable codes
\cite{kerenidis:2004a}, amongst many others.

In this context it is natural to consider a generalisation of the
notion of a boolean function to the quantum domain. While there is
always a temptation to quantise an existing classical concept, there
are other motivations for considering a quantum analogue of the
theory of boolean functions. One reason is to develop lower bounds
for quantum circuits. Another reason is to understand the
propagation of correlations through multipartite quantum systems.
Finally, and perhaps most importantly, we hope that a mature theory
of quantum boolean functions will lead the way to a proof of a
quantum generalisation of the PCP theorem \cite{aharonov:2008a}. We
believe this should occur because the extant classical proofs (see,
eg., \cite{radhakrishnan:2006a, dinur:2007a}) draw heavily on
results from the property testing of boolean functions, which, as
we'll see, generalise naturally to the quantum domain.

In this paper we pursue one particular generalisation of a boolean
function to the quantum domain, namely a unitary operator $f$ which
squares to the identity. We provide several arguments for why our
definition is a natural generalisation of the classical definition,
not least of which are several quantum generalisations of well-known
classical results. These quantum generalisations often require new
proof techniques, and since they reduce in each case to the
classical results when $f$ is diagonal in the computational basis,
we sometimes obtain different proofs for the original classical
results.

%------------------------------------------------------------------------------

\subsection{Summary of results}

The main results we obtain can be summarised as follows.

\begin{itemize}
\item {\bf Quantum property testing.} We give quantum tests that determine whether a unitary operator is a tensor product of Pauli operators, or far from any such tensor product; and similarly whether a unitary operator is a Pauli operator acting on only one qubit, or is far from any such operator. These are quantum generalisations of properties considered in the classical field of property testing of boolean functions. In particular, when applied to classical boolean functions, these tests have better parameters than the original classical tests.

\item {\bf Learning quantum boolean functions.} We develop a quantum analogue of the Goldreich-Levin algorithm, which is an important tool for approximately learning boolean functions \cite{goldreich:1989a,kushilevitz:1993a}. This algorithm allows the approximate learning of quantum dynamics, giving a natural counterpart to recent results of Aaronson on approximately learning quantum states \cite{aaronson:2007a}.

\item {\bf Hypercontractivity and a quantum FKN theorem.} The Friedgut-Kalai-Naor (FKN) theorem \cite{friedgut:2002a} states that boolean functions whose Fourier transform is concentrated on the first level approximately depend on a single variable. We prove a quantum analogue of this statement. In order to obtain this result, we state and prove a quantum generalisation of the hypercontractive inequality of Bonami-Gross-Beckner \cite{bonami:1970a,gross:1975a,beckner:1975a} for functions $\{0,1\}^n \rightarrow \R$. This generalisation may be of independent interest and has several corollaries. Our result is an alternative generalisation of this inequality to that recently proven by Ben-Aroya et al \cite{benaroya:2008a}.

\item {\bf Influences and progress towards a quantum KKL theorem.} The Kahn-Kalai-Linial (KKL) theorem \cite{kahn:1988a} states that every balanced boolean function must have a variable with high influence (qv.). Defining a suitable quantum generalisation of the concept of influence, we prove the generalised theorem in several special cases, and conjecture that it holds in full generality. We also prove a weaker variant (a quantum Poincar\'e inequality).
\end{itemize}
Our presentation is based on the lecture notes \cite{odonnellnotes},
which are an excellent introduction to the field of the analysis of
boolean functions.

%------------------------------------------------------------------------------

\subsection{Related work}

This paper draws heavily on the classical field of the analysis of boolean functions, which for our purposes essentially began with the seminal paper of Kahn, Kalai and Linial \cite{kahn:1988a}, which proved that every balanced boolean function must have an influential variable (see Section \ref{sec:influence}). Since then, a substantial literature has developed, in which statements of interest about boolean functions are proven using mathematical techniques of increasing sophistication, and in particular Fourier analysis. Other important works in this area include a result of Bourgain concerning the Fourier spectrum of boolean functions \cite{bourgain:2002a}, and a result of Friedgut stating that boolean functions which are insensitive on average are close to depending on a small number of variables \cite{friedgut:1998a}.

Ideas relating to the analysis of boolean functions have recently proven fruitful in the study of quantum computation. Indeed, the result of Bernstein and Vazirani giving the first super-polynomial separation between quantum and classical computation \cite{bernstein:1997a} is that quantum computers can distinguish certain boolean functions more efficiently than classical computers can. More recent cases where a quantum advantage is found for classical tasks relating to boolean functions include the quantum Goldreich-Levin algorithm of Adcock and Cleve \cite{adcock:2002a}; the work of Buhrman et al \cite{buhrman:2003a} on quantum property testing of classical boolean functions; and the computational learning algorithms of Bshouty and Jackson \cite{bshouty:1999a}, and also Atici and Servedio \cite{atici:2007a}.

Fourier analysis of boolean functions has been used explicitly to obtain two recent results in quantum computation. The first is an exponential separation between quantum and classical one-way communication complexity proven by Gavinsky et al \cite{gavinsky:2007a}. This separation uses the results of Kahn, Kalai and Linial \cite{kahn:1988a} to lower bound the classical communication complexity of a particular partial function. The second result is a lower bound on the size of quantum random access codes obtained by Ben-Aroya et al \cite{benaroya:2008a}, in which the key technical ingredient is the proof of a matrix-valued extension of the hypercontractive inequality on which \cite{kahn:1988a} is based.

%------------------------------------------------------------------------------
%------------------------------------------------------------------------------

\section{Preliminaries}
In this section we set up our notation and describe the objects we'll
work with in the sequel. We will sometimes use several notations for the same
objects. We've made this decision for two reasons. The first is that
the notations natural in the study of boolean functions are
unnatural in quantum mechanics, and vice versa, and some fluidity
with the notation greatly simplifies the transition between these
two domains. Secondly, two notations means that practitioners in
classical complexity theory and quantum mechanics can (hopefully)
adapt to the tools of the other field.

In the classical domain we work with \emph{boolean functions}
of $n$ variables, which are simply functions $f:\{0, 1\}^n
\rightarrow \{0, 1\}$. We write elements of $\{0,1\}$ as regular
letters, eg.\ $y$, and we write elements of $\{0, 1\}^n$ --
\emph{strings} -- as boldface letters, eg., $\mathbf{x} \equiv
x_1x_2\cdots x_n$, $x_j \in \{0,1\}$, $j = 1, 2, \ldots, n$.
Similarly, we'll write strings in $\{0, 1, 2, 3\}^n$ as boldface
letters starting at $\mathbf{s}$.
It is often convenient to exploit the isomorphism between the set
$\{0, 1\}^n$ of all strings and the set $\mathcal{P}([n])$
of all subsets of $[n]\equiv \{1, 2, \ldots, n\}$ by letting
$\mathbf{x}$ define the subset $S = \{j\in S\,|\, x_j\not=0\}$. We
write subsets of $[n]$ as capital letters beginning with $S$.

There are several elementary functions on $\{0, 1\}^n$ and
$\{0, 1, 2, 3\}^n$ that will be useful
in the sequel. Firstly, we define the \emph{support} of a string
$\mathbf{x}\in\{0,1\}^n$, written $\supp(\mathbf{x})$, via
\begin{equation}
\supp(\mathbf{x}) \equiv \{j\,|\, x_j\not= 0\},
\end{equation}
and for $\mathbf{s}\in\{0,1,2,3\}^n$ via
\begin{equation}
\supp(\mathbf{s}) \equiv \{j\,|\, s_j\not= 0\}.
\end{equation}
We define the \emph{Hamming weight} $|\mathbf{x}|$ of a
string $\mathbf{x}\in\{0,1\}^n$ (respectively,
$\mathbf{s}\in\{0, 1, 2, 3\}^n$) via $|\mathbf{x}| \equiv
|\supp(\mathbf{x})|$ (respectively, $|\mathbf{s}| \equiv
|\supp(\mathbf{s})|$). The \emph{intersection} of two strings
$\mathbf{s}$ and $\mathbf{t}$, written
$\mathbf{s}\cap\mathbf{t}$, is the set $\{i:s_i\neq 0,t_i\neq 0\}$.

We now make the notational switch to
identifying $\{0,1\}$ with $\{+1, -1\}$ via $0 \equiv 1$ and $1 \equiv
-1$. This takes a little getting used to, but is standard in the
boolean function literature; it's often abundantly clear from the
context which notation is being used. Unless otherwise noted, we'll
use the $\{+1, -1\}$ notation.
We identify $\{0, 1\}$ with $\mathbb{Z}/2\mathbb{Z}$, with addition
written $x\oplus y$. We identify $\{+1, -1\}$ with the
multiplicative group of two elements, also isomorphic to
$\mathbb{Z}/2\mathbb{Z}$, and write products as $xy$. This second
identification still allows us to identify elements of the
multiplicative group of two elements with elements of $\mathbb{Z}$
and to use \emph{addition} defined as in $\mathbb{Z}$. Thus we say
that a boolean function is \emph{balanced} if $\sum_{\mathbf{x}} f(\mathbf{x}) =
0$ (here we've made the notational switch).

We denote by $\chi_S$ the \emph{linear function} on subset $S$, defined by
\begin{equation}
\chi_S \equiv \prod_{j\in S}x_j.
\end{equation}
Exploiting the connection between strings and subsets of $[n]$
allows us to write this as
\begin{equation}
\chi_S(\mathbf{x}) \equiv \chi_S(T) \equiv (-1)^{|S\cap T|},
\end{equation}
where $T$ is the set defined by $\mathbf{x}$.

In the quantum domain we work with the Hilbert space of $n$ qubits.
This is the Hilbert space $\mathcal{H} \equiv
(\mathbb{C}^2)^{\otimes n}$. We write elements of $\mathcal{H}$ as
\emph{kets} $|\psi\rangle$, and the inner product between two kets
$|\phi\rangle$ and $|\psi\rangle$ is written $\langle
\phi|\psi\rangle$. There is a distinguished basis for $\mathcal{H}$,
called the \emph{computational basis}, written $|\mathbf{x}\rangle$,
$\mathbf{x}\in \{0, 1\}^n$, with inner product $\langle
\mathbf{x}|\mathbf{y}\rangle =
\delta_{x_1,y_1}\delta_{x_2,y_2}\cdots \delta_{x_n,y_n}$.
We'll also refer to another Hilbert space, namely the Hilbert space
$\mathcal{B}(\mathcal{H})$ of (bounded) operators on $\mathcal{H}$.
The inner product on $\mathcal{B}(\mathcal{H})$ is given by the \emph{Hilbert-Schmidt}
inner product $\ip{M}{N} \equiv \frac{1}{2^n}\tr(M^\dag N)$, $M, N\in
\mathcal{B}(\mathcal{H})$.

We define the (normalised Schatten) $p$-norm of a $d$-dimensional
operator $f$ in terms of the singular values $\{s_j(f)\}$ of $f$
as
\begin{equation}
\|f\|_p \equiv \left( \frac{1}{d} \sum_{j=1}^d s_j(f)^p
\right)^{\frac1p}.
\end{equation}
If we define $|f| \equiv \sqrt{f^\dag f}$ (note the overload of
notation) we can write the $p$-norm as
\begin{equation}
\|f\|_p \equiv \left(\frac{1}{d} \tr(|f|^p)\right)^{\frac1p}.
\end{equation}

Note the useful equalities $\|f^p\|_q = \|f\|_{pq}^p$ and $\|f\|_q^p
= \|f^p\|_{q/p}$. One unfortunate side effect of this normalisation
is that $\|f\|_p$ is not a submultiplicative matrix norm (except at $p=\infty$).
However, we do still have H\"older's inequality \cite{bhatia:1997a}:
for $1/p + 1/q = 1$,
\begin{equation}
 |\ip{f}{g}| \le \|f\|_p \|g\|_q.
\end{equation}

%------------------------------------------------------------------------------
%------------------------------------------------------------------------------

\section{Quantum boolean functions}

The first question we attempt to answer is: what is the right
quantum generalisation of a classical boolean function? We would
expect the concept of a \emph{quantum boolean function} to at least
satisfy the following properties.
\begin{enumerate}
\item A quantum boolean function should be a unitary operator, so it can be implemented
on a quantum computer.
\item Every classical boolean function should give rise to a quantum boolean function in a natural way.
\item Concepts from classical analysis of boolean functions should have natural quantum analogues.
\end{enumerate}
Happily, we can achieve these requirements with the following
definition.
\begin{definition}\label{def:qbf} A quantum boolean function of $n$ qubits is a
unitary operator $f$ on $n$ qubits such that $f^2=\1$.
\end{definition}
We note that it is immediate that $f$ is Hermitian
(indeed, all $f$'s eigenvalues are $\pm 1$), and that every unitary
Hermitian operator is quantum boolean.
Turning to the task of demonstrating the second desired property,
there are at least two natural ways of implementing a classical
boolean function $f: \{0,1\}^n \rightarrow \{0,1\}$ on a quantum
computer. These are as the so-called \emph{bit oracle}
\cite{kashefi:2002a}:
\begin{equation}
|\mathbf{x}\rangle|{y}\rangle \mapsto
U_f|\mathbf{x}\rangle|{y}\rangle \equiv |\mathbf{x}\rangle|{y}\oplus
f(\mathbf{x})\rangle,
\end{equation}
and as what's known as the \emph{phase oracle}:
\begin{equation}
|\mathbf{x}\rangle \mapsto (-1)^{f(\mathbf{x})}|\mathbf{x}\rangle.
\end{equation}
Making the previously discussed notational switch allows us to write the second
action as
\begin{equation}
|\mathbf{x}\rangle \mapsto f(\mathbf{x})|\mathbf{x}\rangle.
\end{equation}
When $f$ acts as a bit oracle it defines a unitary operator $U_f$.
When $f$ acts as a phase oracle, we also obtain a unitary operator,
which we simply write as $f$. As a sanity check of
Definition~\ref{def:qbf} we have that these operators square to the
identity: $U_f^2 = \1$ and $f^2 = \1$, and so are examples of
quantum boolean functions. It turns out that these two actions are
almost equivalent: the phase oracle can be obtained from one use of
the bit oracle and a one-qubit ancilla, via
\be
f(\mathbf{x}) \ket{\mathbf{x}} \frac{1}{\sqrt{2}}(\ket{0}-\ket{1}) \equiv U_f \ket{\mathbf{x}} \frac{1}{\sqrt{2}}(\ket{0}-\ket{1}),
\ee
while the bit oracle can be similarly recovered, given access to a controlled-phase oracle
\be
\ket{\mathbf{x}}\ket{y} \mapsto (-1)^{y \cdot f(\mathbf{x})} \ket{\mathbf{x}}\ket{y},
\ee
where we introduce a one-qubit control register $\ket{y}$.
%
%we recover the phase oracle from two applications of the
%bit oracle via
%\begin{equation}
%f(\mathbf{x})|\mathbf{x}\rangle_A \equiv U_f(\1\otimes
%\sigma^x)U_f|\mathbf{x}\rangle|0\rangle_A,
%\end{equation}
%where the ancilla qubit $A$ is discarded after the application of
%$U_f$.

Having accepted Definition~\ref{def:qbf} for a quantum boolean
function we should at least demonstrate that there is something
quantum about it. This can be easily accomplished by noting that the
theory of quantum boolean functions is at least as general as that
of quantum error correction and quantum algorithms for decision
problems. Firstly, any quantum error correcting code is a subspace
$P$ of $\mathcal{H}$. Using this code we define a quantum boolean
function via $f = \1 - 2P$. Secondly, note that any quantum
algorithm for a decision problem naturally gives rise to a quantum
boolean function in the \emph{Heisenberg picture}: if $U$ is the
quantum circuit in question, and the answer (either $0$ or $1$, or
some superposition thereof) is placed in some output qubit $q$, then
measuring the observable $A = |0\rangle_q\langle0| -
|1\rangle_q\langle1|$ (with an implied action as the identity on the remaining qubits)
will read out the answer. However, this is
equivalent to measuring the observable $f = U^\dag A U$ on the
initial state. It is easy to verify that $f$ is quantum boolean.

%------------------------------------------------------------------------------
%------------------------------------------------------------------------------

\section{Examples of quantum boolean functions}
In this section we introduce several different examples of quantum
boolean functions. We have already met our first such examples,
namely quantisations of classical boolean functions $f(\mathbf{x})$.
Our next example quantum boolean functions are aimed at generalising
various classes of classical boolean functions.

First, we note the obvious fact that even in the single qubit case,
there are infinitely many quantum boolean
functions. By contrast, there are only four one-bit classical
boolean functions.

\begin{example}
For any real $\theta$, the matrix
\be \begin{pmatrix}\cos \theta & \sin \theta \\ \sin \theta & -\cos
\theta \end{pmatrix} \ee
is a quantum boolean function.
\end{example}

Single qubit quantum boolean functions can of course be combined to
form $n$ qubit quantum boolean functions, as follows.

\begin{definition}
Let $f$ be a quantum boolean function. Then we say that $f$ is
\emph{local} if it can be written as
\begin{equation}
f = U_1\otimes U_2 \otimes \cdots \otimes U_n,
\end{equation}
where $U_j$ is a $2\times 2$ matrix satisfying $U_j^2=
\1$, $\forall j \in [n]$.
\end{definition}

Local quantum boolean functions are the natural generalisation
of linear boolean functions (indeed, every linear boolean function is
a local quantum boolean function).
One might reasonably argue that local quantum boolean functions
aren't really quantum: after all, there exists a local rotation
which diagonalises such an operator and reduces it to a classical linear
boolean function. The next example illustrates that the generic
situation is probably very far from this.

\begin{example}
Let $P^2=P$ be a projector. Then $f = \1-2P$ is a quantum boolean
function. In particular, if $|\psi\rangle$ is an arbitrary quantum
state, then $f = \1-2|\psi\rangle\langle \psi|$ is a quantum boolean
function.
\end{example}

If we set $|\psi\rangle = \frac{1}{\sqrt2}(|\mathbf{0}\rangle +
|\mathbf{1}\rangle)$ then we see that there is no local rotation
$V_1\otimes V_2 \otimes \cdots \otimes V_n$ which diagonalises $f =
\1-2|\psi\rangle\langle \psi|$.

%------------------------------------------------------------------------------

\subsection{New quantum boolean functions from old}

There are several ways to obtain new quantum boolean functions from
existing quantum boolean functions. We summarise these constructions
below.

\begin{lemma}
Let $f$ be a quantum boolean function and $U$ be a unitary operator.
Then $U^\dag f U$ is a quantum boolean function.
\end{lemma}

\begin{lemma}
Let $h$ be a Hermitian operator. Then $f=\sgn(h)$ is a quantum
boolean function, where
\begin{equation}
\sgn(x) = \begin{cases} 1, &\quad x>0 \\ -1, &\quad x \le 0,
\end{cases}
\end{equation}
and as usual the notation $f=g(h)$, for a Hermitian operator $h$, means the operator obtained by applying the function $g:\R \rightarrow \R$ to the eigenvalues of $h$.
\end{lemma}

\begin{lemma}
\label{lem:anticommute}
Let $\alpha_j\in\mathbb{R}$, $j \in [m]$, satisfy $\sum_{j=1}^m
\alpha_j^2 = 1$ and $\{f_j\}$, $j \in [m]$, be a set of
\emph{anticommuting} quantum boolean functions, i.e.\
\begin{equation}
\{f_j,f_k\} \equiv f_jf_k+f_kf_j = 2\delta_{jk}\1.
\end{equation}
Then
\begin{equation}
f = \sum_{j=1}^m \alpha_j f_j
\end{equation}
is a quantum boolean function.
\end{lemma}
\begin{proof}
Squaring $f$ gives
\begin{equation}
f^2 = \sum_{j=1}^m \alpha_j^2 f_j^2 + \sum_{j<k} \alpha_j\alpha_k
\{f_j,f_k\} = \1.
\end{equation}
That $f$ is Hermitian follows because it is a real-linear combination of Hermitian operators. 
\end{proof}

%------------------------------------------------------------------------------
%------------------------------------------------------------------------------

\section{Fourier analysis}
It is well-known that every function $f:\{0,1\}^n \rightarrow
\mathbb{R}$ can be expanded in terms of the characters of the group
$(\mathbb{Z}/2\mathbb{Z})^n$. These characters are given by
the set of linear functions $\chi_{S}(T) = (-1)^{|S\cap T|}$ (we are
identifying input strings $\mathbf{x}$ with the subset $T$). This
expansion is called the Fourier transform over
$(\mathbb{Z}/2\mathbb{Z})^n$.

The use of Fourier analysis in the study of boolean functions was
pioneered by Kahn, Kalai, and Linial, who were
responsible for the eponymous KKL theorem \cite{kahn:1988a}, and has facilitated many
of the core foundational results relating to boolean functions. We seek an analogous
expansion for quantum boolean functions.

Our quantum analogues of the characters of $\mathbb{Z}/2\mathbb{Z}$
will be the Pauli matrices $\{\sigma^0, \sigma^1,
\sigma^2, \sigma^3\}$\footnote{The Pauli matrices are often written
as $\sigma^0 \equiv \1$, $\sigma^1 \equiv \sigma^x$,
$\sigma^2 \equiv \sigma^y$, and $\sigma^3 \equiv \sigma^z$.}. These are defined as
\begin{equation}
\sigma^{0} = \begin{pmatrix}1 & 0 \\ 0 & 1\end{pmatrix}, \;
\sigma^{1} = \begin{pmatrix}0 & 1 \\ 1 & 0\end{pmatrix}, \;
\sigma^{2} = \begin{pmatrix}0 & -i \\ i & 0\end{pmatrix}, \;
\text{and}\; \sigma^{3} = \begin{pmatrix}1 & 0 \\ 0 &
-1\end{pmatrix}.
\end{equation}
It is clear that the Pauli operators are quantum boolean functions.
A tensor product of Paulis (also known as a \emph{stabilizer operator}) is
written as $\sigma^{\mathbf{s}}\equiv\sigma^{s_1}\otimes\sigma^{s_2}\otimes
\cdots \otimes \sigma^{s_n}$, where $s_j\in \{0, 1, 2, 3\}$. We use the
notation $\sigma^j_i$ for the operator which acts as $\sigma^j$ at the $i$'th
position, and trivially elsewhere.

Where convenient we'll use one
of two other different notations to refer to the operators
$\sigma^{\mathbf{s}}$, namely as $\chi_{\mathbf{s}}$ and as
$|\mathbf{s}\rangle$. We use the first of these alternate notations
to bring out the parallels between the classical theory of boolean functions and
the quantum theory, and the second to emphasise the Hilbert space
structure of $\mathcal{B}(\mathcal{H})$. We write the inner products
in each case as
\begin{equation}
\langle \mathbf{s}|\mathbf{t}\rangle \equiv \frac{1}{2^n}
\tr((\sigma^{\mathbf{s}})^\dag\sigma^{\mathbf{t}}) \equiv
\ip{\chi_{\mathbf{s}}}{\chi_{\mathbf{t}}}.
\end{equation}

The set of stabilizer operators is orthonormal with respect to the Hilbert-Schmidt inner product, and thus forms a basis for the vector space $\mathcal{B}(\mathcal{H})$. So we can express any (bounded) operator $f$ on $n$ qubits in terms of stabilizer operators, with the explicit expansion
\begin{equation}
f = \sum_{\mathbf{s}\in\{0, 1, 2, 3\}^n}
\fhats \chi_{\mathbf{s}},
\end{equation}
where $\fhats = \ip{\chi_{\mathbf{s}}}{f}$.  In an
abuse of terminology, we call the set $\{\fhats\}$
indexed by $\mathbf{s}\in\{0, 1, 2, 3\}^n$ the
\emph{Fourier coefficients} of $f$. Note that, if
$f$ is Hermitian, these coefficients are all real.
This expansion is well-known in quantum information theory
and, for example, was recently used by Kempe et al \cite{kempe:2008a}
to give upper bounds on fault-tolerance thresholds.

We extend the definition of support to operators via
\begin{equation}
\supp(f) \equiv \bigcup_{\mathbf{s}\,|\, \fhats\neq 0}
\supp(\mathbf{s}).
\end{equation}

Similarly, we define the \emph{weight} $\wgt(M)$ of an operator $M$
via $\wgt(M) \equiv |\supp(M)|$.

\begin{definition}
Let $f$ be a quantum boolean function. If $|\supp(f)| = k$ then
we say that $f$ is a \emph{$k$-junta}. If $k=1$ then we say that $f$
is a \emph{dictator}. If $k=0$ then we say that $f$ is \emph{constant}.
\end{definition}

We note that while, classically, there is a distinction between dictators
(functions $f:\{0,1\}^n \rightarrow \{0,1\}$ such that $f(\mathbf{x}) = x_i$
for some $i$) and so-called {\em anti-dictators} (functions where $f(\mathbf{x}) = 1-x_i$
for some $i$), there is no such distinction in our terminology here.

Our definition of $\{\fhats \}$ as the Fourier
coefficients for an operator $f$ is no accident. Indeed, it turns
out that when one expands a \emph{boolean function} $f$ (represented
as a phase oracle) in terms of $\chi_{\mathbf{s}}$ then we recover
the classical Fourier transform of $f$.
\begin{proposition}
Let $f$ be a boolean function $f:\{0,1\}^n \rightarrow \{1,-1\}$. Then if $f$ acts as
$f|\mathbf{x}\rangle = f(\mathbf{x})|\mathbf{x}\rangle$, and if
\begin{equation}
f = \sum_{\mathbf{s}} \fhats \chi_{\mathbf{s}},
\end{equation}
then the set $\{\fhats\}$ are given by
\begin{equation}
\fhats =
\begin{cases} 0, &\quad \mathbf{s}\not\in\{0,3\}^n,\\
\frac{1}{2^n}\sum_{T\subset [n]}(-1)^{|S\cap T|}f(T),&\quad
\mathbf{s}\in\{0,3\}^n,\end{cases}
\end{equation}
where $S = \supp(\mathbf{s})$ and $T$ is the support of the string
$\mathbf{t} = t_1t_2\cdots t_n$, where $t_j = 0$ if $j\not\in T$ and
$t_j = 3$ if $j\in T$.
\end{proposition}
\begin{proof}
Because $f(\mathbf{x})$ is diagonal in the computational basis we
immediately learn that $\ip{\chi_{\mathbf{s}}}{f} = 0$ for any
$\mathbf{s}\not\in \{0,3\}^n$. When
$\mathbf{s}\in\{0,3\}^n$ we have that
\begin{equation}
\begin{split}
\fhats &= \ip{\chi_{\mathbf{s}}}{f}%\\
=
\frac{1}{2^n}\tr\left(f\prod_{j\in\supp(\mathbf{s})}\sigma^3_{j}\right)
\\
&= \frac{1}{2^n}\sum_{\mathbf{x}}
\left(\prod_{j\in\supp(\mathbf{s})}(-1)^{x_j}\right)f(\mathbf{x})%\\
= \frac{1}{2^n}\sum_{T\subset [n]} (-1)^{|S\cap T|}f(T).
\end{split}
\end{equation}
\end{proof}

We also immediately have the quantum analogues of Plancherel's theorem
and Parseval's equality.

\begin{proposition}
Let $f$ and $g$ be operators on $n$ qubits. Then $\ip{f}{g} =
\sum_{\mathbf{s}} \fhats^* \ghats$. Moreover, $\|f\|_2^2 = \sum_{\mathbf{s}}
|\fhats|^2$. Thus, if $f$ is quantum boolean,
$\sum_{\mathbf{s}} \fhats^2 = 1$.
\end{proposition}

It is often convenient to decompose the Fourier expansion of an operator into different levels. An arbitrary $n$-qubit operator $f$ can be expanded as
\be
f = \sum_{\mathbf{s}} \fhats \chis = \sum_{k=0}^n f^{=k},
\ee
where
\be f^{=k} \equiv \sum_{\mathbf{s},|\mathbf{s}|=k} \fhats \chis. \ee
(One can define $f^{<k}$, etc., similarly.) The {\em weight} of $f$ at level $k$ is then defined as $\|f^{=k}\|_2^2$. A natural measure of the complexity of $f$ is its {\em degree}, which is defined as
\be \mathrm{deg}(f) \equiv \max_{\mathbf{s},\fhats\neq 0}
|\mathbf{s}|. \ee
We pause to note an important difference between quantum and classical boolean
functions. In the classical case, it is easy to show that every non-zero Fourier coefficient of a
boolean function on $n$ bits is at least $2^{1-n}$ in absolute value. However, this does not
hold for quantum boolean functions. Consider the operator
\be \label{eq:smallcoeff}
f = \epsilon\,\sigma^1 \otimes \sigma^1 + \sqrt{1-\epsilon^2}\,\sigma^2 \otimes \1. \ee
By Lemma \ref{lem:anticommute}, this is a quantum boolean function for any $0 \le \epsilon \le 1$.
Taking $\epsilon \rightarrow 0$, we see that the coefficient of $\sigma^1 \otimes \sigma^1$
may be arbitrarily small.

%------------------------------------------------------------------------------
%------------------------------------------------------------------------------

\section{Testing quantum boolean functions}

The field of classical {\em property testing} solves problems of the
following kind. Given access to a boolean function $f$ that is promised
to either have some property, or to be ``far'' from having some property,
determine which is the case, using a small number of queries to $f$.
Property testers are an important component of many results in classical
computer science, e.g.\ \cite{hastad:2001a,dinur:2007a}; see the review
\cite{fischer:2001a} for an introduction to the area.

In this section we describe property testing for {\em quantum}
boolean functions: we give quantum algorithms which
determine whether a unitary operator, implemented as an
oracle, has the property of being, variously, a stabilizer operator,
or both a stabilizer operator and a dictator.
These tests differ substantially from their classical counterparts
and typically require fewer queries. However, as with their
classical equivalents, we use Fourier analysis to bound their
probabilities of success.

We note that Buhrman et al \cite{buhrman:2003a} have already shown that
quantum computers can obtain an advantage over classical computers for the
property testing of {\em classical} boolean functions.

%------------------------------------------------------------------------------

\subsection{Closeness}

The tests that we define will, informally, output either that a
unitary operator has some property, or is
``far'' from any operator with that property. In order to define
the concept of property testing of quantum boolean functions, we
thus need to define what it means for two operators to be close.

\begin{definition}
Let $f$ and $g$ be two operators. Then we say that
$f$ and $g$ are $\epsilon$-close if $\|f-g\|_2^2 \le 4
\epsilon$.
\end{definition}
In quantum theory, it is often natural to use the infinity, or
\emph{sup}, norm to measure closeness of operators (i.e., the
magnitude of the largest eigenvalue). However, the 2-norm seems more
intuitive when dealing with boolean functions; for example, if we
produce a pair of quantum boolean functions $f$, $g$ from {\em any}
classical boolean functions that differ at any position, then
$\|f-g\|_\infty = 2$. Intuitively, the 2-norm tells us how the
function behaves \emph{on average}, and the infinity norm tells us
about the \emph{worst case} behaviour. There is also the following
relationship to unitary operator discrimination.

\begin{proposition}
Given a unitary operator $f$ promised to be one of two unitary
operators $f_1$, $f_2$, there is a procedure that determines whether
$f=f_1$ or $f=f_2$ with one use of $f$ and success probability
\begin{equation}
\frac{1}{2} + \frac{1}{2}\sqrt{1-|\ip{f_1}{f_2}|^2}.
\end{equation}
\end{proposition}

\begin{proof}
The proof rests on the fundamental result of Holevo and Helstrom
\cite{holevo:1973, helstrom:1976} which says that the exact
minimum probability of error that can be achieved when
discriminating between two pure states $\ket{\psi_1}$ and $\ket{\psi_2}$ with a
priori probabilities $p$ and $1-p$ is given by
\begin{equation}
\mathbb{P}[\text{test succeeds}] = \frac12+\frac12\sqrt{1-4p(1-p)|\vip{\psi_1}{\psi_2}|^2}.
\end{equation}

We now apply $f_1$ and $f_2$ to halves of two maximally entangled
states $|\Phi\rangle \equiv \frac{1}{2^{n/2}}\sum_{\mathbf{x}}
|\mathbf{x}\rangle|\mathbf{x}\rangle$. This produces the states
\begin{equation}
\begin{split}
|f_1\rangle \equiv f_1\otimes \1 |\Phi\rangle, \quad
\text{and}\quad |f_2\rangle \equiv f_2\otimes \1
|\Phi\rangle.
\end{split}
\end{equation}
The overlap between these two states can be calculated as follows:
\begin{equation}
\begin{split}
\langle f_1|f_2\rangle &= \frac{1}{2^n} \sum_{\mathbf{x},\mathbf{y}} \langle\mathbf{y}|\langle\mathbf{y}| (f_1^\dag f_2)\otimes \1 |\mathbf{x}\rangle|\mathbf{x}\rangle\\
&= \frac{1}{2^n} \sum_{\mathbf{x},\mathbf{y}} \langle\mathbf{y}|f_1^\dag f_2|\mathbf{x}\rangle \langle\mathbf{y}|\mathbf{x}\rangle\\
&=\frac{1}{2^n} \sum_{\mathbf{x}} \langle\mathbf{x}|f_1^\dag
f_2|\mathbf{x}\rangle = \frac{1}{2^n}\tr(f_1^\dag f_2) = \ip{f_1}{f_2}.
\end{split}
\end{equation}

The lemma follows when we apply the Holevo-Helstrom result to $\ket{f_1}$
and $\ket{f_2}$ and minimise over $p$.
\end{proof}

%------------------------------------------------------------------------------

\subsection{The quantum stabilizer test}
In this subsection we describe a quantum test, the \emph{quantum
stabilizer test}, which decides, using only two queries, whether a
unitary operator $f$ is either $\epsilon$-close to a stabilizer operator $\chi_{\mathbf{s}}$
up to a phase, or is far from any such operator.
We also describe a test which is conjectured to
decide whether a unitary operator is $\epsilon$-close to
local or not.

The idea behind our tests is very simple: suppose $f$ is a
unitary operator, and we want to work out if it is a stabilizer,
i.e.\ if $f = \chi_{\mathbf{s}}$ for some $\mathbf{s}$. One way to do this is to apply $f$
to the halves of $n$ maximally entangled states resulting in a
quantum state $f\otimes\mathbb{I}|\Phi\rangle$. If $f$ is local
then the result will just be a tensor product of $n$ (possibly
rotated) maximally entangled states, and if $f$ is a stabilizer then
it should be an $n$-fold product of one of four possible states. If
not, then there will be entanglement between the $n$ subsystems. The
way to test this hypothesis is to create another identical state
$f\otimes \mathbb{I}|\Phi\rangle$ by again applying $f$ to another
set of $n$ maximally entangled states and separately apply an
equality test to each of the $n$ subsystems which are meant to be
disentangled from each other.

\begin{definition}
Let $f$ be a unitary operator on $n$ qubits. The quantum stabilizer and
locality tests proceed as follows.
\begin{enumerate}
\item Prepare $4n$ quantum registers in the state
\begin{equation}
|\Phi\rangle_{\mathbf{A}{\mathbf{A}}'}|\Phi\rangle_{\mathbf{B}{\mathbf{B}}'}
\equiv \frac{1}{2^{n}}\sum_{\mathbf{x}, \mathbf{y}}
|\mathbf{x}\rangle_{\mathbf{A}}|\mathbf{x}\rangle_{{\mathbf{A}}'}|\mathbf{y}\rangle_{\mathbf{B}}|\mathbf{y}\rangle_{{\mathbf{B}'}}
= |\phi\rangle^{\otimes
n}_{\mathbf{A}{\mathbf{A}'}}\otimes|\phi\rangle_{\mathbf{B}{\mathbf{B}'}}^{\otimes
n},
\end{equation}
where $|\phi\rangle\equiv
\frac{1}{\sqrt{2}}\sum_{x\in\{\pm1\}} |xx\rangle$, and $\mathbf{A} =
A_1A_2\cdots A_n$, ${\mathbf{A}'} =
{A_1'}{A_2'}\cdots {A_n'}$, $\mathbf{B} =
B_1B_2\cdots B_n$, and ${\mathbf{B}'} =
{B_1'}{B_2'}\cdots {B_n'}$.
\item Apply $f$ to $\mathbf{A}$ and once more to $\mathbf{B}$ to
give
\begin{equation}
|f\rangle|f\rangle = f_{\mathbf{A}}\otimes
\1_{{\mathbf{A}'}}\otimes f_{\mathbf{B}}\otimes
\1_{{\mathbf{B}'}}|\Phi\rangle|\Phi\rangle.
\end{equation}
\item To test if $f$ is a stabilizer measure the \emph{equality} observable\footnote{This is not quite the same as measuring both subsystems and checking if the result is equal, as that would destroy coherent superpositions like $1/\sqrt{2}(|00\rangle + |11\rangle)$ which would otherwise be left undisturbed by measurement of {\sc eq}.}:
\begin{equation}
\text{{\sc eq}} = \left(\sum_{s=0}^3
|s\rangle_{A{A'}}\langle s|\otimes |s\rangle_{B{B'}}
\langle s|\right)^{\otimes n},
\end{equation}
where
\begin{equation}
|s\rangle \equiv \chi_s\otimes \1|\phi\rangle, \quad
s\in\{0, 1, 2, 3\},
\end{equation}
giving
\begin{equation}
\mathbb{P}[\text{test accepts}] = \langle f| \langle f|\text{{\sc
eq}}|f \rangle |f\rangle.
\end{equation}
\item To test locality measure the \emph{swap observable}
\begin{equation}
\text{{\sc sw}} = \frac{1}{2^n}\left(\1_{A{A'}:B{B'}}+\text{\sc
swap}_{A{A'}:B{B'}}\right)^{\otimes n},
\end{equation}
where $\text{\sc swap}_{X:Y}$ is the operator that swaps the
two subsystems $X$ and $Y$. This gives
\begin{equation}
\mathbb{P}[\text{test accepts}]  = \langle f| \langle
f|\text{{\sc sw}}|f \rangle |f\rangle.
\end{equation}
\end{enumerate}
\end{definition}

We now prove the following
\begin{proposition}
Suppose that a unitary operator $f$ passes the quantum
stabilizer test with probability $1-\epsilon$, where $\epsilon<1/2$. Then $f$ is
$\epsilon$-close to an operator $e^{i\phi}\chi_{\mathbf{s}}$,
$\mathbf{s}\in\{0, 1, 2, 3\}^n$, for some phase $\phi$. If $f$ is a stabilizer operator then it passes the stabilizer test with probability 1.
\end{proposition}
\begin{proof}
Expand $f$ in the Fourier basis as
\begin{equation}\label{eq:opexp}
f = \sum_{\mathbf{s}} \fhats \chi_{\mathbf{s}}.
\end{equation}
Noting that 
\begin{equation*}
	\langle f| \langle f|\text{{\sc eq}}|f \rangle |f\rangle = \sum_{\mathbf{s}, \mathbf{t}} |\fhats|^2|\hat{f}_{\mathbf{t}}|^2\langle \mathbf{tt}|\mathbf{ss}\rangle,
\end{equation*}
we see that 
\begin{equation}
\mathbb{P}[\text{test accepts}]  = \sum_{\mathbf{s}}|\fhats|^4.
\end{equation}
Now, thanks to Parseval's relation, we have that
\begin{equation}
1 = \sum_{\mathbf{s}} |\fhats|^2,
\end{equation}
and, given that the test passes with probability $1-\epsilon$, we
thus have
\begin{equation}
1-\epsilon \le \sum_{\mathbf{s}} |\fhats|^4 \le
\left(\max_{\mathbf{s}}|\fhats|^2\right)\sum_{\mathbf{s}}
|\fhats|^2 = \max_{\mathbf{s}}|\fhats|^2.
\end{equation}
So, according to Parseval, there is exactly one term $|\fhats|^2$ in
the expansion (\ref{eq:opexp}) which is at least $1-\epsilon$,
and the rest are each at most $\epsilon$. Thus $f$ is
$\epsilon$-close to $e^{i \phi} \chi_{\mathbf{s}}$ for some phase $\phi$
(we have that $|\langle f,\chi_{\mathbf{s}}\rangle| \ge \sqrt{1-\epsilon}$).
\end{proof}

\begin{remark}
	The stabilizer test is a quantum generalisation of the classical \emph{linearity test} of Blum, Luby, and Rubenfeld \cite{blum:1993a} (the BLR test). When interpreted in quantum language the BLR test can be seen as a method to test if a quantum boolean function diagonal in the computational basis is close to a tensor product of $\sigma^3$s. (A task for which the stabilizer test can also be applied.) It is notable that the BLR test requires \emph{three} queries to $f$ in order to achieve the same success probability as its quantum counterpart, which only requires \emph{two} queries. Thus, the stabilizer test can be said to have better parameters than its classical counterpart.
\end{remark}

Now we turn to the quantum locality test.
\begin{proposition}
The probability that the quantum locality test accepts when applied to an operator $f$ is equal to
\begin{equation}
\frac{1}{2^n}\left(\sum_{S\subset [n]} \tr(\rho_S^2)\right),
\end{equation}
where $\rho_S$ is the partial trace of $|f\rangle\langle f|$ over
all subsystems $A_j{A'}_j$ with $j\not\in S$, and we define
$\tr(\rho_{\emptyset}^2) = 1$.
\end{proposition}

\begin{proof}
The proof proceeds via direct calculation:
\begin{equation}
\begin{split}
\mathbb{P}[\text{\rm test accepts}] &= \langle f|\langle
f|\left(\frac{\1_{A{A'}:B{B'}}+\text{\sc
swap}_{A{A'}:B{B'}}}{2}\right)^{\otimes n}
|f\rangle|f\rangle \\
&= \frac{1}{2^n}\sum_{S\subset [n]}\langle f|\langle f|\text{\sc
swap}_{A_S{A'}_S:B_S{B'}_S}|f\rangle|f\rangle \\
&= \frac{1}{2^n}\sum_{S\subset [n]}\tr(\text{\sc
swap}_{A_S{A'}_S:B_S{B'}_S} \rho_S \otimes \rho_S ) \\
&= \frac{1}{2^n}\sum_{S\subset [n]}\tr(\rho_S^2),
\end{split}
\end{equation}
where we use the notation $A_S$ for an operator which is applied to the
subsystems $S$, and acts as the identity elsewhere.
\end{proof}

It is easy to see that if $f$ is local then the probability the
quantum locality test accepts is equal to $1$. We have been unable
to show that if the test accepts with probability greater than
$1-\epsilon$ then $f$ is close to being local. Thus we have the
following
\begin{conjecture}
\label{con:locality}
Suppose $f$ passes the quantum locality test with probability $\ge
1-\epsilon$. Then there exist $U_j$, $j\in [n]$, with $U_j^2 =
\mathbb{I}$ such that
\begin{equation}
\langle f, U_1\otimes U_2 \otimes \cdots \otimes U_n\rangle \ge
1-2\epsilon.
\end{equation}
\end{conjecture}

%------------------------------------------------------------------------------

\subsection{Testing dictators}
In this subsection we describe a quantum test --- a quantisation of
the H{\aa}stad test \cite{hastad:2001a} --- which tests whether a
unitary operator $f$ is $\epsilon$-close to a dictator. In fact,
we give two such tests. The first determines whether $f$ is close, up to a phase, to
a dictator which is also a stabilizer operator (a {\em stabilizer dictator}).
The second is intended to determine whether $f$ is close to a dictator.
As with the situation for quantum locality testing, we are able to analyse
the first test, but leave the second as a conjecture.

The dictator --- or \emph{quantum H{\aa}stad} --- test is defined as follows.
\begin{definition}
Let $f$ be a unitary operator and let $0\le \delta \le 1$.
Then the quantum H{\aa}stad test proceeds as follows.
\begin{enumerate}
\item Prepare $4n$ quantum registers in the state
\begin{equation}
|\Phi\rangle_{\mathbf{A}{\mathbf{A}'}}|\Phi\rangle_{\mathbf{B}{\mathbf{B}'}}.
\end{equation}
(Our notation is identical to that of the quantum locality test.)
\item Apply $f$ to $\mathbf{A}$ and once more to $\mathbf{B}$ to
give
\begin{equation}
|f\rangle|f\rangle = f_{\mathbf{A}}\otimes
\1_{{\mathbf{A}'}}\otimes f_{\mathbf{B}}\otimes
\1_{{\mathbf{B}'}}|\Phi\rangle|\Phi\rangle.
\end{equation}
\item To test if $f$ is close to a stabilizer dictator measure the \emph{$\delta$-equality} POVM given by the operators $\{\text{{\sc eq}}_\delta, \1 - \text{{\sc eq}}_\delta\}$, where
\begin{equation}
\text{{\sc eq}}_\delta = \left(|00\rangle\langle 00| +
(1-\delta)\sum_{s=1}^3 |s\rangle_{A{A'}}\langle s|\otimes
|s\rangle_{B{B'}} \langle s|\right)^{\otimes n}.
\end{equation}
This measurement is easy to implement by flipping a $\delta$-biased coin, and gives
\begin{equation}
\mathbb{P}[\text{test accepts}] = \langle f| \langle f|\text{{\sc
eq}}_\delta|f \rangle |f\rangle.
\end{equation}
\item To test if $f$ is a dictator measure the \emph{$\delta$-swap observable}
\begin{equation}
\text{{\sc sw}}_\delta =
\frac{1}{2^n}\left(T(\delta)_{A{A'}:B{B'}}+\sqrt{T(\delta)}\text{\sc
swap}\sqrt{T(\delta)}_{A{A'}:B{B'}}\right)^{\otimes
n},
\end{equation}
where
\begin{equation}
T(\delta) = \sum_{s,t} (1-\delta)^{\frac{|s|+|t|}{2}}
|s,t\rangle\langle s, t|
\end{equation}
(recall that $|s| = 1$ if $s>0$, for $s\in\{0,1,2,3\}$), giving
\begin{equation}
\mathbb{P}[\text{test accepts}]  = \langle f| \langle
f|\text{{\sc sw}}|f \rangle |f\rangle.
\end{equation}
\end{enumerate}
\end{definition}

We now prove the following
\begin{proposition}
Suppose that a unitary operator $f$ passes the stabilizer
H{\aa}stad test with $\delta = \frac{3}{4}\epsilon$ with probability
$1-\epsilon$. Then $f$ is $\epsilon$-close, up to a phase, to $\1$ or a stabilizer
dictator. (We assume $\epsilon \le 0.01$.)
\end{proposition}
\begin{proof}
Write the Fourier expansion of $f$ as follows:
\begin{equation}
f = \sum_{\mathbf{s}} \fhats {\chi}_{\mathbf{s}}.
\end{equation}
It is easy to verify that
\begin{equation}
\mathbb{P}[\text{test accepts}] = \sum_{\mathbf{s}}
(1-\delta)^{|\mathbf{s}|}|\fhats|^4.
\end{equation}
Now suppose that $f$ is a stabilizer dictator, up to a phase, on some variable $j\in [n]$,
i.e., $f$ is $e^{i \phi} \sigma_j^s$, for some $s\in\{0,1,2,3\}$. Then $f$
passes the quantum H{\aa}stad test with probability
\begin{equation}
\mathbb{P}[\text{test accepts}] = (1-\delta)|\fhats|^4 = 1-\delta.
\end{equation}
On the other hand, suppose that
\begin{equation}
1-\epsilon \le \mathbb{P}[\text{test accepts}] \le \left(\sum_{\mathbf{s}}
|\fhats|^2\right)\max_{\mathbf{s}} (1-\delta)^{|\mathbf{s}|}|\fhats|^2 =
\max_{\mathbf{s}}
\left(1-\frac{3}{4}\epsilon\right)^{|\mathbf{s}|}|\fhats|^2.
\end{equation}
Since $(1-\delta)^{|\mathbf{s}|} \le 1$ it follows that there exists
some $\mathbf{s}$ such that $|\fhats|^2 \ge 1-\epsilon$.  Using the
fact that $(1-\frac{3}{4}\epsilon)^{|\mathbf{s}|} < 1-\epsilon$
for $\supp(\mathbf{s}) \ge 2$, we know that this maximum occurs on a
string $\mathbf{s}$ with support one or zero. That is, there exists
a Fourier coefficient of magnitude at least $1-\epsilon$ on a string
with support at most one.
\end{proof}

We have been unable to prove the corresponding result for the full
quantum dictator test, so we leave this as a conjecture.
\begin{conjecture}
\label{con:dictator}
Suppose that a unitary operator $f$ passes the quantum
H{\aa}stad dictator test with $\delta = \frac{3}{4}\epsilon$ with
probability $1-\epsilon$. Then $f$ is $\epsilon$-close, up to a phase, to $\1$
or a dictator. (Assume $\epsilon \le 0.01$.)
\end{conjecture}

%------------------------------------------------------------------------------
%------------------------------------------------------------------------------

\section{Learning quantum boolean functions}

The purpose of this section is to describe a family of results in
the spirit of the \emph{Goldreich-Levin algorithm}
\cite{goldreich:1989a}, an algorithm which was originally defined in
a cryptographic context, but was shown by Kushilevitz and Mansour \cite{kushilevitz:1993a}
to be a useful tool for learning boolean functions. Continuing
the theme of the previous sections, we'll see that quantum computers
are polynomially more efficient at learning tasks for boolean
functions. Heuristically, this is because quantum computers can
exploit quantum superposition to carry out ``super-dense'' coding
\cite{nielsen:2000a}, allowing us to pack more information in a
single quantum query.

The presentation of the results in this section is based on \cite{odonnellnotes}.

\subsection{Learning stabilizer operators and approximating Fourier coefficients}

We begin by learning the class of stabilizer operators. It turns out
that this can be done with only one quantum query, in contrast to the $n$ queries required
classically. This is a natural generalisation of the quantum algorithm
of Bernstein and Vazirani \cite{bernstein:1997a}, which learns linear
boolean functions with one quantum query. For simplicity, the results in this subsection are
given in terms of quantum boolean functions, but it should be clear how to extend them
to general unitary operators.
\begin{proposition}
If a quantum boolean function $f$ is a stabilizer operator, then we
can identify $f$ with $1$ quantum query, using $O(n)$ quantum measurements of Pauli operators.
\end{proposition}
\begin{proof}
The idea behind the proof is simple: we apply $f$ to one half of a
collection of $n$ maximally entangled states and then measure in a
basis of maximally entangled states. More precisely, suppose that $f
= \chi_{\mathbf{s}}$ for some $\mathbf{s}$. Then the first step of
our algorithm queries $f$ to produce the state (our notation is
identical to that of the previous section)
\begin{equation}
f\otimes \mathbb{I}|\Phi\rangle_{\mathbf{A}{\mathbf{A}'}}
\equiv |\mathbf{s}\rangle.
\end{equation}
Since the set of states $\{|\mathbf{s}\rangle\}$ indexed by
$\mathbf{s}$ forms an orthonormal basis for
$\mathbf{A}{\mathbf{A}'}$ we can simply measure the state
$|\mathbf{s}\rangle$ to find out $\mathbf{s}$. (The $O(n)$ measurements bit
comes from the preparation step and the measurement step; one needs
to measure each register $A_j{A'}_j$ separately.)
\end{proof}

The next proposition shows us that the previous result is robust
against perturbations.
\begin{proposition}\label{prop:robust}
Suppose that a quantum boolean function $f$ satisfies
\begin{equation}\label{eq:nearlin}
\fhats \ge \frac{1+\epsilon}{\sqrt{2}}
\end{equation}
for some $\mathbf{s}$. Then $\chi_{\mathbf{s}}$ can be identified with probability
$1-\delta$ with
$O\left(\frac{1}{\epsilon^2}\log\left(\frac{1}{\delta}\right)\right)$
uses of $f$.
\end{proposition}

\begin{proof}
Note that, by Parseval, there can only be one character
$\chi_\mathbf{s}$ which satisfies (\ref{eq:nearlin}); the rest of
the characters must be further from $f$ than $\chi_{\mathbf{s}}$.

The strategy of our proof is simple: we make $q$ queries to $f$ by
applying it to sets of maximally entangled states $|\Phi\rangle$ and
then measure each resulting state in the $\{|\mathbf{s}\rangle\}$
basis. We then take a majority vote. For each query, with
probability $\mathbb{P}[\text{succ}] \ge \frac12 + \epsilon$, we get
the right answer. To work out the probability that the test fails we
bound the failure probability by bounding the cumulative
distribution function of the binomial distribution $B(q,p)$ with $p
= \frac12 + \epsilon$:
\begin{equation}
\begin{split}
\mathbb{P}[\text{test fails}] = \mathbb{P}[\text{at least $q/2$
failures}] &\le e^{-2\frac{(qp-\frac{q}{2})^2}{q}} \\
&= e^{-2q\epsilon^2} = \delta,
\end{split}
\end{equation}
so that choosing $q =
O\left(\frac{1}{\epsilon^2}\log\left(\frac{1}{\delta}\right)\right)$
gives us the desired result.
\end{proof}

\begin{lemma}\label{lem:fourchunks}
Let $f = \sum_{\mathbf{s}\in \{0,1,2,3\}^{n}} \hat{f}_{\mathbf{s}}
\chi_{\mathbf{s}}$ be a quantum boolean function. Then
$\hat{f}_{\mathbf{s}}^2 \ge \gamma^2$ for at most
$\frac{1}{\gamma^2}$ terms.
\end{lemma}
\begin{proof}
This is a simple consequence of Parseval's relation.
\end{proof}

\begin{lemma}\label{lem:fcoeffest}
For any $\mathbf{s}\in\{0,1,2,3\}^{n}$ it is possible to estimate
$\hat{f}_{\mathbf{s}}$ to within $\pm \eta$ with probability
$1-\delta$ using
$O\left(\frac{1}{\eta^2}\log\left(\frac{1}{\delta}\right)\right)$
queries.
\end{lemma}
\begin{proof}
To prove this lemma we need access to the controlled-$f$ quantum
boolean function $U_f = |0\rangle_C\langle 0|\otimes \mathbb{I}_A +
|1\rangle_C \langle1|\otimes f_A$. (This can be easily implemented
using $f$ alone by adjoining $n$ ancilla qubits and pre-applying a
controlled-{\sc swap} operation between the main qubits and ancilla
qubits, applying $f$ to the ancillas and post-applying a
controlled-{\sc swap} operation and then discarding the ancillas.)

The method takes place on a register consisting of the system
$\mathbf{A}$ and a copy of the system ${\mathbf{A}'}$ and a
control qubit. It proceeds as follows.
\begin{enumerate}
\item Prepare the control+system+copy in
$|0\rangle_C|\Phi\rangle_{\mathbf{A}{\mathbf{A}'}}$.
\item Apply a Hadamard operation $H = \frac{1}{\sqrt2}\left(\begin{smallmatrix} 1 & 1 \\ 1 & -1\end{smallmatrix}\right)$ to $C$: the system is now in the state
\begin{equation}
\frac{1}{\sqrt{2}}|0\rangle_C|\Phi\rangle_{\mathbf{A}{\mathbf{A}'}}
+
\frac{1}{\sqrt{2}}|1\rangle_C|\Phi\rangle_{\mathbf{A}{\mathbf{A}'}}.
\end{equation}
\item Apply the controlled-$\chi_{\mathbf{s}}$ operation
 $V_{\chi_{\mathbf{s}}}
= |0\rangle_C\langle0|\otimes \chi_{\mathbf{s}}+ |1\rangle_C\langle
1|\otimes \mathbb{I}$ (implemented in the same way as $U_f$, above)
to yield
\begin{equation}
\frac{1}{\sqrt{2}}|0\rangle_C\chi_{\mathbf{s}}|\Phi\rangle_{\mathbf{A}{\mathbf{A}'}}
+
\frac{1}{\sqrt{2}}|1\rangle_C|\Phi\rangle_{\mathbf{A}{\mathbf{A}'}}.
\end{equation}
\item Next apply $U_f$ to give
\begin{equation}
\frac{1}{\sqrt{2}}|0\rangle_C\chi_{\mathbf{s}}|\Phi\rangle_{\mathbf{A}{\mathbf{A}'}}
+
\frac{1}{\sqrt{2}}|1\rangle_Cf|\Phi\rangle_{\mathbf{A}{\mathbf{A}'}}.
\end{equation}
\item Apply a Hadamard operation once again to $C$. The system is now in the
state
\begin{equation}
\frac{1}{2}|0\rangle_C(\chi_{\mathbf{s}}|\Phi\rangle_{\mathbf{A}{\mathbf{A}'}}+f|\Phi\rangle_{\mathbf{A}{\mathbf{A}'}})
+
\frac{1}{2}|1\rangle_C(\chi_{\mathbf{s}}|\Phi\rangle_{\mathbf{A}{\mathbf{A}'}}-f|\Phi\rangle_{\mathbf{A}{\mathbf{A}'}}).
\end{equation}
\item Now measure the control qubit in the computational basis. This
gives ``0'' with probability
\begin{equation}
\mathbb{P}[0] = \langle\Phi|\frac{(\chi_{\mathbf{s}} +
f)^2}{4}|\Phi\rangle = \frac12 +
\frac{1}{2\cdot2^n}\tr(\chi_{\mathbf{s}} f) = \frac12 + \frac{1}{2}
\hat{f}_{\mathbf{s}}
\end{equation}
and ``1'' with probability
\begin{equation}
\mathbb{P}[1] = \langle\Phi|\frac{(\chi_{\mathbf{s}} -
f)^2}{4}|\Phi\rangle = \frac12 -
\frac{1}{2\cdot2^n}\tr(\chi_{\mathbf{s}} f) = \frac12 - \frac{1}{2}
\hat{f}_{\mathbf{s}}.
\end{equation}
\end{enumerate}
An application of Hoeffding's inequality yields the desired result.
\end{proof}

\begin{remark}
	Note that one may improve the performance of the procedures involved the proofs of Proposition~\ref{prop:robust}  and Lemma~\ref{lem:fcoeffest} by exploiting quantum amplitude amplification. This achieves a square-root improvement of the dependence on $\epsilon$ and $\eta$, respectively.
\end{remark}

\subsection{The quantum Goldreich-Levin algorithm}
In this subsection we describe the quantum Goldreich-Levin
algorithm.

\begin{theorem}[quantum Goldreich-Levin]\label{thm:qgl}
Given oracle access to a unitary operator $f$ on $n$ qubits and its adjoint $f^\dag$, and given
$\gamma, \delta > 0$, there is a $\poly\left(n,
\frac{1}{\gamma}\right)\log\left(\frac{1}{\delta}\right)$-time
algorithm which outputs a list $L = \{\mathbf{s}_1, \mathbf{s}_2,
\ldots, \mathbf{s}_m\}$ such that with probability $1-\delta$: (1)
if $|\hat{f}_{\mathbf{s}}| \ge \gamma$, then $\mathbf{s} \in L$; and
(2) for all $\mathbf{s} \in L$, $|\hat{f}_{\mathbf{s}}| \ge \gamma/2$.
\end{theorem}

This quantum algorithm can be understood as a kind of branch and
bound algorithm: we initially assume that the set $S_n$ of all $4^n$
strings $\mathbf{s}$ contributes significantly to the Fourier
expansion of $f$. We then partition this set into four equal chunks
and efficiently estimate (via Proposition~\ref{prop:weightest}) the
total weight of the Fourier expansion on each of these chunks. We
then throw away the chunks with low weight and repeat the process by
successively partitioning the remaining chunks into four, etc. The
reason the total number of remaining chunks doesn't blow
up exponentially is because of Lemma~\ref{lem:fourchunks}.

\begin{definition}
Let $f$ be a unitary operator on $n$ qubits. Let $I\subset [n]$. For any
$\mathbf{s}\in\{0,1,2,3\}^{|I|}$ with $S \equiv
\supp(\mathbf{s})\subset I$ define
\begin{equation}
F_{\mathbf{s}; I} \equiv
\frac{1}{2^{|I|}}\tr_{I}(\chi_{\mathbf{s}}f).
\end{equation}
\end{definition}

\begin{lemma}
Let $f$ be a unitary operator on $n$ qubits, then
\begin{equation}
\frac{1}{2^{n-|I|}}\tr(\chi_{\mathbf{t}}F_{\mathbf{s}; I}) =
\hat{f}_{\mathbf{s}\cup\mathbf{t}},
\end{equation}
for any $\mathbf{t}\in\{0,1,2,3\}^{n-|I|}$ with $T\equiv
\supp(\mathbf{t})\subset I^c$, where $I^c$ denotes the complement of the set $I$
and $\mathbf{s}\cup\mathbf{t}$ denotes concatenation, i.e.,
\begin{equation}
[\mathbf{s}\cup\mathbf{t}]_j =
\begin{cases}s_j, &\quad j\in I \\ t_j, &\quad j\not\in
I.\end{cases}
\end{equation}
\end{lemma}
\begin{proof}
Note that
\begin{equation}
\begin{split}
F_{\mathbf{s}; I} &= \frac{1}{2^{n-|I|}}\sum_{\supp(\mathbf{t})\subset I^c} \tr(\chi_{\mathbf{t}} F_{\mathbf{s};I})\chi_{\mathbf{t}}  \\
&= \frac{1}{2^{n}}\sum_{\supp(\mathbf{t})\subset I^c}
\tr(\chi_{\mathbf{s}}\otimes\chi_{\mathbf{t}}f)\chi_{\mathbf{t}},
\end{split}
\end{equation}
from which the result follows.
\end{proof}

\begin{lemma}
Let $f$ be a unitary operator on $n$ qubits. Then
\begin{equation}
\frac{1}{2^{n-|I|}}\tr(F_{\mathbf{s}; I}^\dag F_{\mathbf{s}; I}) =
\sum_{\mathbf{t} \,|\, t_j = s_j, \forall j\in I} |\hat{f}_{\mathbf{t}}|^2.
\end{equation}
\end{lemma}

\begin{proof}
Consider
\begin{equation}
\frac{1}{2^{n-|I|}}\tr(F_{\mathbf{s}; I}^\dag F_{\mathbf{s}; I}) =
\sum_{\supp(\mathbf{u})\subset I^c}
|\hat{f}_{\mathbf{s}\cup\mathbf{u}}|^2 = \sum_{\mathbf{t} \,|\, t_j =
s_j, j\in I} |\hat{f}_{\mathbf{t}}|^2.
\end{equation}
\end{proof}

It is convenient to write the set $\mathcal{S}$ of all $\mathbf{t}$
such that $t_j = s_j$, $j\in I$ as an \emph{indicator string}:
\begin{definition}
Let $I\subset [n]$. The indicator string $\mathcal{S} = (s_{j_1},
s_{j_2}, \ldots, s_{j_{|I|}},
*, *, \ldots, *)$, where $s_{j_k} \in \{0, 1, 2, 3\}$, is the set
\begin{equation}
\{\mathbf{t} \,|\, t_{j_k} = s_{j_k}, j_k\in I\}.
\end{equation}
\end{definition}

\begin{definition}
The \emph{weight} $W(\mathcal{S})$ of an indicator string
$\mathcal{S}$ is
\begin{equation}
W(\mathcal{S}) = \sum_{\mathbf{t} \in \mathcal{S}}
|\fhatt|^2.
\end{equation}
\end{definition}

It turns out that we can efficiently estimate $W(\mathcal{S})$.
\begin{proposition}\label{prop:weightest}
Let $f$ be a unitary operator on $n$ qubits. Then for any indicator string
$\mathcal{S}$ it is possible to efficiently estimate the weight
$W(\mathcal{S})$ to within $\pm \gamma^2$ with probability
$1-\delta$ using
$O\left(\frac{1}{\gamma^4}\log\left(\frac1{\delta}\right)\right)$
queries to $f$ and $f^\dag$.
\end{proposition}

\begin{proof}
Our method takes place in a system consisting of one ancilla qubit
$C$, four copies, called $A_1$, $A_1'$, $A_2$, and $A_2'$, of the
qubits in $I$ and two copies, called $B$ and $B'$, of the qubits in
$I^c$. Thus the total system is $CA_1A_1'A_2A_2'BB'$. The algorithm
proceeds as follows.
\begin{enumerate}
\item Initialise the system (by applying a Hadamard on C) into
\begin{equation}
\frac{1}{\sqrt{2}} (|0\rangle +|1\rangle)_C
|\Phi\rangle_{A_1A_1'}|\Phi\rangle_{A_2A_2'}|\Phi\rangle_{BB'}.
\end{equation}
(From now on, for simplicity, we write $|\boldsymbol{\Phi}\rangle \equiv
|\Phi\rangle_{A_1A_1'}|\Phi\rangle_{A_2A_2'}|\Phi\rangle_{BB'}$.)
\item Apply the operation $U \equiv |0\rangle_C\langle0|\otimes\chi_{\mathbf{s}}\otimes\chi_{\mathbf{s}} + |1\rangle_C\langle1|\otimes \mathbb{I}_{A_1A_2}$ on $CA_1A_2$. The system is now in the state
\begin{equation}
\frac{1}{\sqrt{2}} |0\rangle_C (\chi_{\mathbf{s}}\otimes
\chi_{\mathbf{s}}|\boldsymbol{\Phi}\rangle) + \frac{1}{\sqrt{2}}
|1\rangle_C|\boldsymbol{\Phi}\rangle.
\end{equation}
\item Apply the controlled-$(f,f^\dag)$ operation $V \equiv
|0\rangle_C\langle0|\otimes \mathbb{I}_{A_1A_2B} +
|1\rangle_C\langle1|\otimes f_{A_1B}f_{A_2B}^\dag$ on $CA_1A_2B$. (This
operation is easy to implement with two applications of
controlled-$f$ and controlled-$f^\dag$ operations.) The system is now in the state
\begin{equation}
\frac{1}{\sqrt{2}} |0\rangle_C (\chi_{\mathbf{s}}\otimes
\chi_{\mathbf{s}}|\boldsymbol{\Phi}\rangle) + \frac{1}{\sqrt{2}}
|1\rangle_C(f_{A_1B}f_{A_2B}^\dag|\boldsymbol{\Phi}\rangle).
\end{equation}
\item Apply a Hadamard to the control register:
\begin{equation}
\frac{1}{2} |0\rangle_C (\chi_{\mathbf{s}}\otimes \chi_{\mathbf{s}}
+ f_{A_1B}f_{A_2B}^\dag)|\boldsymbol{\Phi}\rangle + \frac{1}{2} |1\rangle_C
(\chi_{\mathbf{s}}\otimes \chi_{\mathbf{s}} -
f_{A_1B}f_{A_2B}^\dag)|\boldsymbol{\Phi}\rangle.
\end{equation}
\item Measure the control register in the computational basis. We
get ``0'' with probability:
\begin{equation}
\begin{split}
\mathbb{P}[0] &= \frac14\langle \boldsymbol{\Phi}|(\chi_{\mathbf{s}}\otimes
\chi_{\mathbf{s}} + f_{A_2B} f_{A_1B}^\dag)(\chi_{\mathbf{s}}\otimes
\chi_{\mathbf{s}}
+ f_{A_1B} f_{A_2B}^\dag)|\boldsymbol{\Phi}\rangle\\
&= \frac{1}{2} +
\frac{1}{2}\frac{1}{2^{n+|I|}}\mbox{Re}(\tr(\chi_{\mathbf{s}}\otimes
\chi_{\mathbf{s}}f_{A_1B}f_{A_2B}^\dag)) \\
&= \frac{1}{2} + \frac{1}{2}
\left(\frac{1}{2^{n-|I|}}\tr(F_{\mathbf{s}; I}^\dag F_{\mathbf{s};
I})\right) = \frac12 +\frac12 W(\mathcal{S}),
\end{split}
\end{equation}
with a similar formula for $\mathbb{P}[1]$; the second equality follows from noting that $\langle \boldsymbol{\Phi}| M_{A_1A_2B}\otimes \mathbb{I}_{A_1'A_2'B'} | \boldsymbol{\Phi}\rangle = \frac{1}{2^{n+|I|}}\tr(M)$, for any operator $M$ on $A_1A_2B$ (the extra $|I|$ comes from the fact that $A$ appears twice). An application of
Hoeffding's inequality gives the desired result.
\end{enumerate}
\end{proof}

We now describe the quantum Goldreich-Levin algorithm.

\begin{algorithm}[H]
\caption{Quantum Goldreich-Levin algorithm} \label{alg:qgl}
\begin{algorithmic}
\STATE $L \leftarrow (*,*, \ldots, *)$ %
\FOR{$k=1$ to $n$} %
\FOR{each $\mathcal{S}\in L$, $\mathcal{S} = (s_1,
s_2, \ldots, s_{k-1}, *, *, \ldots, *)$} %
\STATE Let $\mathcal{S}_{s_k} = (s_1, s_2, \ldots, s_{k-1}, s_k,
*, *, \ldots,
*)$ then for $s_k = 0, 1, 2, 3$ estimate $W(\mathcal{S}_{s_k})$ to within $\pm
\gamma^2/4$ with probability at least $1-\delta$. %
\STATE Remove $\mathcal{S}$ from $L$. %
\STATE Add $\mathcal{S}_{s_k}$ if the estimate of
$W(\mathcal{S}_{s_k})$ is at least $\gamma^2/2$ for $s_k = 0, 1, 2,
3$.
\ENDFOR %
\ENDFOR
\end{algorithmic}
\end{algorithm}

We now analyse the algorithm. To simplify the analysis we'll assume
that all estimations are accurate. We'll remove this assumption
later.

\begin{lemma}
After 1 iteration of the algorithm, $W(\mathcal{S}) \ge
\frac{\gamma^2}{4}$ for all $\mathcal{S}\in L$.
\end{lemma}
\begin{proof}
All the estimates are assumed to be correct, and for all
$\mathcal{S}\in L$, $\mathcal{S}$ was entered into the list $L$
because its estimated weight was at least $\frac{\gamma^2}{2}$, and
the estimate is correct to within an additive $\frac{\gamma^2}{4}$.
\end{proof}

\begin{lemma}
At any time $|L| \le \frac{4}{\gamma^2}$.
\end{lemma}
\begin{proof}
The result follows from Lemma~\ref{lem:fourchunks}.
\end{proof}

\begin{lemma}
The quantum Goldreich-Levin algorithm requires at most a total of
$\frac{16n}{\gamma^2}$ estimations.
\end{lemma}
\begin{proof}
At each iteration there are at most $\frac{4}{\gamma^2}$ items in
the list and the algorithm needs at most $4$ estimations per
iteration. There are only $n$ iterations.
\end{proof}

\begin{lemma}
For any $\mathbf{s}$ such that $|\hat{f}_{\mathbf{s}}| \ge \gamma$,
there exists $\mathcal{S}\in L$ such that $\mathbf{s}\in
\mathcal{S}$.
\end{lemma}
\begin{proof}
	It suffices to note that $|\hat{f}_{\mathbf{s}}|^2 \ge \gamma^2$, thus the weight of any string $\mathcal{S}$ containing $\mathbf{s}$ is greater than $\gamma^2$, hence at least once such $\mathcal{S}$ remains in the list $L$ after the first step.
\end{proof}

To remove the accuracy assumption, given $\delta > 0$, we define
$\delta' = \frac{\delta \gamma^2}{16n}$, and perform each estimation
with confidence $1-\delta'$. By the union bound, if the algorithm
performs $\frac{16n}{\gamma^2}$ estimations, they are all correct
with probability at least $1-\delta$, so the algorithm is correct
with probability at least $1-\delta$.

We now have all the ingredients to prove Theorem~\ref{thm:qgl}:
\begin{proof}[Proof of Theorem~\ref{thm:qgl}]

The total running time is dominated by the estimations. There are at
most $\frac{16n}{\gamma^2}$, and each takes
$O(\log(\frac{1}{\delta})/\gamma^2)$ samples to estimate, so the
overall running time is
$\mbox{poly}(n,\frac1\delta)\log(\frac1\delta)$.

\end{proof}
\subsection{Learning quantum dynamics}
In this subsection we show how to apply the quantum Goldreich-Levin
algorithm to learn the dynamics generated by local quantum systems.
In principle one needs an exponential number of queries to learn the
dynamics associated with a quantum system of $n$ qubits. However, if
we make the key physical assumption that the dynamics are generated
by a \emph{geometrically local} quantum system, then it turns out
that we can do much better. We'll focus, for simplicity, on
\emph{one-dimensional} quantum systems, but our results extend
pretty easily to higher dimensional systems and to any system which
is local on a graph with bounded isoperimetric dimension. There are
even some results available for quantum dynamics on general graphs,
which we leave to future works.

Before we begin we'll provide some background on quantum dynamics.
Let $\mathcal{H} \cong \C(2^{n})$ be the Hilbert space associated with a
collection of $n$ qubits. A \emph{Hamiltonian} is a Hermitian
operator $H$ on $\mathcal{H}$. The \emph{dynamics generated by $H$}
is the one-parameter family of unitary operators $U(t) = e^{itH}$.

We now need to describe what it means to \emph{learn} the dynamics
generated by a Hamiltonian $H$:
\begin{definition}
Let $U$ be a unitary operator (not necessarily a quantum boolean
function). We say that we have \emph{$(\gamma, \epsilon)$-learned}
the \emph{dynamics} of a \emph{known} Hermitian operator $M$ if, given $\gamma$ queries
of $U$, we can provide an estimate $\widetilde{U^\dag M U}$ of
$U^\dag M U$ such that $\|\widetilde{U^\dag M U} - U^\dag M U\|_2^2
\le \epsilon$ with probability greater than $1-\delta$.
\end{definition}

\begin{definition}
A \emph{one-dimensional quantum Hamiltonian} $H$ is any Hamiltonian
which can be written
\begin{equation}
H = \sum_{j=1}^{n-1} h_j
\end{equation}
with $h_j$ Hermitian, $\|h_j\|_{\infty} = O(1)$, and $\supp(h_j)
\subset \{j, j+1\}$ for $j\le n-1$.
\end{definition}

The key to our main result is the following estimate.
\begin{theorem}[Lieb-Robinson bound \cite{lieb:1972a}]
Let $H$ be a one-dimensional quantum system. Then for all $s$ and $j$
\begin{equation}
\|e^{-itH}\sigma_j^{s} e^{itH} - e^{-itH_{\Lambda}}\sigma_j^{s}
e^{itH_{\Lambda}}\|_\infty \le ce^{k|t|-v|\Lambda|},
\end{equation}
where $c$, $k$, and $v$ are constants independent of $n$, $\Lambda \subset
[n-1]$ is any contiguous subset centred on $j$, and $H_{\Lambda} = \sum_{j\in\Lambda} h_j$.
\end{theorem}

We can use the Lieb-Robinson bound to establish the following
corollary
\begin{corollary}
Let $H$ be a one-dimensional quantum system. Then
\begin{equation}
\frac{1}{2^n}\|e^{-itH}\sigma_j^s e^{itH} -
e^{-itH_{\Lambda}}\sigma_j^s e^{itH_{\Lambda}}\|_2^2 \le
ce^{k|t|-v|\Lambda|},
\end{equation}
with possibly new constants $c$, $k$, and $v$.
\end{corollary}
\begin{proof}
The result is a simple application of the matrix norm inequality
\begin{equation}
\|M\|_2^2 \le m\|M\|^2_\infty
\end{equation}
for Hermitian $m\times m$ matrices $M$.
\end{proof}

\begin{proposition}
Let $t= O(\log(n))$. Then, with probability $1-\delta$ we can
$(\gamma, \epsilon)$-learn the quantum boolean functions
$\sigma_j^{s}(t)\equiv e^{-itH}\sigma_j^{s} e^{itH}$ using $\gamma
= \poly(n,1/\epsilon,\log(1/\delta))$ queries of $e^{itH}$. (Note that all queries are made to the unitary operator $U = e^{itH}$ and the pauli operators are assumed to be not evolving during the execution of the algorithm.)
\end{proposition}

\begin{proof}
The Lieb-Robinson bound tells us that if $t= O(\log(n))$ then the
only significant Fourier coefficients of the
quantum boolean function $\sigma_j^{s}(t)$ are those whose support is centred on $j$ and have size 
the same order as $|t|$. Since there are at most $O(n4^{|t|}) =
O(\poly(n))$ such coefficients we can \emph{efficiently} apply the
quantum Goldreich-Levin algorithm to output a list of them. Given
this list we then use Lemma~\ref{lem:fcoeffest} to individually
estimate them to accuracy $\epsilon$. An application of the union
bound gives us the result.
\end{proof}

\begin{remark}
The operators $\sigma_j^{s}(t)$ are significant in condensed matter
physics as they represent the dynamics of applied magnetic fields.

Although it is straightforward to compute an approximation to
$\langle \psi|\sigma_j^{s}(t)|\psi\rangle$ for a \emph{given}
initial quantum state $|\psi\rangle$, when we
$(\gamma,\epsilon)$-learn $\sigma_j^{s}(t)$ we are demanding that we can, \emph{on
average}, calculate a good approximation to $\langle
\psi|\sigma_j^{s}(t)|\psi\rangle$ for arbitrary quantum states which
are randomly chosen \emph{after} the algorithm has terminated.
Notice also that the algorithm we've presented here doesn't need to
know which qubit interacts with which, just that the qubits interact
in a line. Indeed, simple modifications of the Lieb-Robinson bound
allow us to conclude that we can efficiently learn the dynamics of
qubits which interact on graphs with, eg., a finite number of
randomly placed bridges.
\end{remark}

%------------------------------------------------------------------------------
%------------------------------------------------------------------------------

\section{Noise and quantum hypercontractivity}

One of the most useful tools in the classical analysis of boolean
functions has been the hypercontractive (also known as
``Bonami-Gross-Beckner'') inequality \cite{beckner:1975a, bonami:1968a,
bonami:1970a, gross:1975a, nelson:1965a, nelson:1973a, rudin:1960a}.
Perhaps the most intuitive way to write down this inequality is in
terms of a \emph{noise operator} on functions, which can be defined
in two equivalent ways.

\begin{definition}
For a bit-string $x \in \{0,1\}^n$, define the distribution $y \sim_\epsilon x$ as follows: each bit of $y$ is equal to the corresponding bit of $x$ with probability $1/2+\epsilon/2$, and flipped with probability $1/2-\epsilon/2$. Then the noise operator with rate $-1 \le \epsilon \le 1$, written $T_\epsilon$, is defined
via
\begin{equation}
\label{eqn:smoothing}
(T_\epsilon f)(x) = \E_{y \sim_\epsilon x}[f(y)].
\end{equation}
Equivalently, $T_\epsilon$ may be defined by its action on Fourier coefficients, as follows.
\begin{equation}
T_\epsilon f = \sum_{\mathbf{s}\in\{0,1\}^n} \epsilon^{|\mathbf{s}|}
\fhats \chi_{\mathbf{s}}.
\end{equation}
\end{definition}

It is easy to see that noise rate $1$ leaves the function as it is, whereas noise rate
$0$ replaces the function with the constant function $\hat{f}_{\mathbf 0} \1$.
The Fourier-analytic definition given above immediately extends to the quantum setting, giving
a superoperator defined as follows.

\begin{definition}
The noise superoperator with rate\footnote{This restriction on
$\epsilon$ is necessary for the map to be completely positive
\cite{king:2001a}.} $-1/3 \le \epsilon \le 1$, written $T_\epsilon$,
is defined as follows.
\begin{equation}
T_\epsilon f = \sum_{\mathbf{s}\in\{0,1,2,3\}^n} \epsilon^{|\mathbf{s}|}
\fhats \chi_{\mathbf{s}}.
\end{equation}
\end{definition}

Perhaps surprisingly, just as the classical noise operator has an
alternative natural definition in terms of ``local smoothing'' (eqn.\ (\ref{eqn:smoothing})), its
quantum generalisation has a natural definition too, in terms of the action
of the qubit depolarising channel.

\begin{proposition}
Let $f$ be a Hermitian operator on $n$ qubits. Then, for $-1/3 \le \epsilon \le 1$,
$T_\epsilon f = \mathcal{D}^{\otimes n}_\epsilon f$, where $\mathcal{D}_\epsilon$ is the
\emph{qubit depolarising channel} with noise rate $\epsilon$, i.e.\
\begin{equation}\label{eq:depnoise}
\mathcal{D}_\epsilon (f) = \frac{(1-\epsilon)}{2} \tr (f) \1 + \epsilon
f.
\end{equation}
\end{proposition}

\begin{proof}
To verify the claim it is only necessary to check
(\ref{eq:depnoise}) on the characters $\chi_{\mathbf{s}}$ and then
extend by linearity. Thus, since for a single qubit
$\mathcal{D}_{\epsilon}(\chi_{0}) = \chi_{0}$, and
$\mathcal{D}_{\epsilon}(\chi_{j}) = \epsilon\chi_{j}$ for $j\not=0$ we have
that
\begin{equation}
\mathcal{D}^{\otimes n}_\epsilon (\chi_{\mathbf{s}}) =
\epsilon^{|\mathbf{s}|} \chis = T_\epsilon(\chi_{\mathbf{s}}).
\end{equation}
\end{proof}
We have the following easy observations about the behaviour of the noise superoperator.

\begin{itemize}
\item $\|T_\epsilon f\|_p \le \|f\|_p$ for any $0\le \epsilon \le 1$. Follows because $T_\epsilon f$ is a convex combination of conjugations by unitaries.
\item The semigroup property $T_\delta T_\epsilon = T_{\delta \epsilon}$ is immediate.
\item If $\delta \le \epsilon$, $\|T_\delta f\|_p \le \|T_\epsilon f\|_p$ for any $p\ge 1$. Follows from the previous two properties.
\item For any $f$ and $g$, $\ip{T_\epsilon f}{g} = \ip{f}{T_\epsilon g}$. Follows from Plancherel's theorem: $\ip{T_\epsilon f}{g}=\sum_{\mathbf{s}} \widehat{T_\epsilon f}_\mathbf{s}^*\,\ghats = \sum_{\mathbf{s}} \epsilon^{|\mathbf{s}|}\fhats^* \ghats = \sum_{\mathbf{s}} \fhats^* \widehat{T_\epsilon g}_\mathbf{s} = \ip{f}{T_\epsilon g}$.
\end{itemize}

We are now ready to state the quantum hypercontractive inequality. By the identification of the noise superoperator with a tensor product of qubit depolarising channels, this inequality is really a statement about the properties of this channel, and can also be seen as a generalisation of a hypercontractive inequality of Carlen and Lieb \cite{carlen:1993a}.

\begin{theorem}\label{thm:hypercontractivity}
Let $f$ be a Hermitian operator on $n$ qubits and assume that $1 \le p \le 2 \le q \le \infty$. Then, provided that
\begin{equation}
\epsilon \le \sqrt{\frac{p-1}{q-1}},
\end{equation}
we have
\begin{equation}
\|T_\epsilon f\|_q \le \|f\|_p.
\end{equation}
\end{theorem}

Just as in the classical case, the proof of this inequality will
involve two steps: first, a proof for $n=1$, and then an inductive
step to extend to all $n$. The proof of the base case is essentially the same as in the classical setting.
In the classical proof, the inductive step uses a quite general tensor product argument. One might hope that this argument extended to the quantum setting. This would be true if it held that, for any channels (superoperators) $C$ and $D$ on $n$ qubits,
\be \| C \otimes D \|_{q \rightarrow p} \le \| C \|_{q \rightarrow
p} \|D\|_{q \rightarrow p}, \ee
where we define the $q\rightarrow p$ norm as
\be \| C \|_{q \rightarrow p} = \sup_f \frac{\| C f \|_p}{\|f\|_q}.
\ee
However, this most general statement is actually {\em false}, as it would imply the so-called ``maximal output $p$-norm multiplicativity conjecture'' in the case $q=1$, to which counterexamples have recently been found by Hayden and Winter \cite{hayden:2008b} for all $p>1$ and for $p=1$ by Hastings \cite{hastings:2009a}. Our proof must therefore be specific to the depolarising channel, and will turn out to rely on a non-commutative generalisation of Hanner's inequality recently proven by King \cite{king:2003a}, rather than the Minkowski inequality used in the classical proof. In fact, a corollary of our result is the proof of multiplicativity of the maximum output $q \rightarrow p$ norm for the depolarising channel for certain values of $q$ and $p$.

Before we begin the proof of Theorem~\ref{thm:hypercontractivity} in
earnest, we will need some subsidiary lemmata.

\begin{lemma}
\label{lem:basecase}
Let $f$ be a single qubit Hermitian operator and let $1 \le p \le q \le \infty$. Then, provided that
\begin{equation}
\epsilon \le \sqrt{\frac{p-1}{q-1}}
\end{equation}
we have
\begin{equation}\label{eq:1qubithc}
\|T_\epsilon f\|_q \le \|f\|_p.
\end{equation}
\end{lemma}
\begin{proof}
We diagonalise $f$ as $U^\dag \Lambda U$, where $\Lambda$ is diagonal. Since the depolarising
channel is symmetric we have that $T_\epsilon(U^\dag \Lambda U) =
\mathcal{D}_\epsilon(U^\dag \Lambda U) = U^\dag
\mathcal{D}_\epsilon(\Lambda) U = U^\dag T_\epsilon(\Lambda) U$ and since
the $q$-norm is unitarily invariant we have that $\|T_\epsilon(f)\|_q =
\|T_\epsilon(\Lambda)\|_q$. So we may as well focus on diagonal
operators $\Lambda = \left(\begin{smallmatrix} \lambda_1 & 0 \\ 0 &
\lambda_2\end{smallmatrix}\right)$. In terms of the eigenvalues of
$\Lambda$ the inequality (\ref{eq:1qubithc}) is just the
two-point inequality established by Bonami \cite{bonami:1970a},
Gross \cite{gross:1975a} and Beckner \cite{beckner:1975a}.
\end{proof}

\begin{lemma}
\label{lem:p=2}
It suffices to prove Theorem \ref{thm:hypercontractivity} for $p=2$.
\end{lemma}

\begin{proof}
The classical proof (see, e.g., \cite[Lecture 12]{mosselnotes}) goes through unchanged. The first step is to prove that Theorem \ref{thm:hypercontractivity} holds for $q=2$, and any $1 \le p \le 2$. Assume the theorem holds for $p=2$ and any $q$, and take $p'$ such that $1/p' + 1/p = 1$; then
\bea
\|T_{\sqrt{p-1}} f\|_2 &=& \sup_{\|g\|_2=1} |\ip{g}{T_{\sqrt{p-1}}\,f}| = \sup_{\|g\|_2=1} |\ip{T_{\sqrt{p-1}}\,g}{f}| \\
&\le& \|f\|_p \sup_{\|g\|_2=1} \|T_{\sqrt{p-1}}\,g\|_{p'} \le \|f\|_p,
\eea
where the first inequality is H\"older's inequality. Now, to prove the theorem when $1 \le p < 2 < q$, we use the semigroup property:
\be \|T_{\sqrt{\frac{p-1}{q-1}}} f\|_q = \|T_{1/\sqrt{q-1}}\,T_{\sqrt{p-1}}\,f\|_q \le \|T_{\sqrt{p-1}}\,f\|_2 \le \|f\|_p. \ee
\end{proof}

\begin{lemma}
\label{lem:positive} For all $p\ge 1$, $\|T_\epsilon f\|_p \le
\|T_\epsilon |f|\|_p$, where $|f|$ is the operator $\sqrt{f^2}$.
\end{lemma}

\begin{proof}
This holds because $f \le |f|$ (in a positive semidefinite sense)
and applying $T_\epsilon$ doesn't change this ordering because it is
a convex combination of conjugations by unitaries.
\end{proof}

We are now ready to prove Theorem \ref{thm:hypercontractivity}. By Lemma \ref{lem:p=2}, it suffices to prove the following statement.

\begin{proposition}
Let $f$ be a Hermitian operator. Then $\|T_\epsilon f
\|_{1+1/\epsilon^2} \le \|f\|_2$ for any $0\le \epsilon \le 1$.
\end{proposition}

\begin{proof}
For readability, it will be convenient to switch to the un-normalised standard
Schatten $p$-norm, so for the remainder of the proof $\|f\|_p =
\left(\tr |f|^p\right)^{1/p}$. With this normalisation,
what we want to prove is
\be \|T_\epsilon f \|_{1+1/\epsilon^2} \le
2^{-n\left(\frac{1-\epsilon^2}{2(1+\epsilon^2)}\right)} \|f\|_2. \ee
The proof is by induction on $n$. The theorem is true for $n=1$ by
Lemma \ref{lem:basecase}, so assume $n>1$ and expand $f$ as follows.
\be f = \1 \otimes a + \sigma^1 \otimes b + \sigma^2 \otimes c +
\sigma^3 \otimes d  =
\begin{pmatrix}a+d&b-ic\\b+ic&a-d\end{pmatrix}.\ee
Then, by direct calculation, we have
\be T_\epsilon f = \begin{pmatrix} T_\epsilon(a + \epsilon d) &
\epsilon\, T_\epsilon(b-ic)\\ \epsilon\, T_\epsilon(b+ic) &
T_\epsilon(a-\epsilon d) \end{pmatrix}, \ee
%We consider $\| T_\epsilon f \|_p^p$, for some $p$ to be later determined.
%
where the operator $T_\epsilon$ on the left-hand side acts on $n$ qubits, while
the operator $T_\epsilon$ on the right-hand side acts on $n-1$ qubits.
For brevity, set $q=1+1/\epsilon^2$, and assume that $\|T_\epsilon f\|_q^q \le 2^{-(n-1)(q/2-1)} \|f\|_2^q$
for any $(n-1)$-qubit Hermitian operator $f$.
Using a non-commutative Hanner's inequality for positive block
matrices \cite{king:2003a}, which holds for $q \ge 2$, and noting that we can assume that $f$ is positive by
Lemma \ref{lem:positive}, we obtain
\beas
\| T_\epsilon f \|_q^q &\le& \left\|\begin{pmatrix}\|T_\epsilon(a + \epsilon d)\|_q & \|\epsilon T_\epsilon(b-ic)\|_q\\ \|\epsilon T_\epsilon(b+ic)\|_q & \|T_\epsilon(a - \epsilon d)\|_q \end{pmatrix}\right\|_q^q\\
&\le& 2^{-(n-1)(q/2-1)} \left\|\begin{pmatrix}\|a + \epsilon d\|_2 & \epsilon\|b-ic\|_2 \\ \epsilon\|b+ic\|_2 & \|a - \epsilon d\|_2 \end{pmatrix}\right\|_q^q,
\eeas
where we use the inductive hypothesis. The proposition will thus be
proven if we can show that
\beas
\|g\|_q^q &\equiv& \left\|\begin{pmatrix}\|a + \epsilon d\|_2 & \epsilon\|b-ic\|_2 \\ \epsilon\|b+ic\|_2 & \|a - \epsilon d\|_2 \end{pmatrix}\right\|_q^q\\
&\le& 2^{-(q/2-1)} \|f\|_2^q = 2\left( \|a\|_2^2 + \|d\|_2^2 + \|b-ic\|_2^2 \right)^{q/2},
\eeas
where we call the matrix on the left-hand side of the inequality
$g$. Write $g = T_\epsilon h$ for some $h$, where the elements of $h$ are
\beas
h_{11} &=& \frac{1}{2}\left(\left(1+\frac{1}{\epsilon}\right)\|a+\epsilon d\|_2 + \left(1-\frac{1}{\epsilon}\right)\|a-\epsilon d\|_2 \right),\\
h_{12} &=& \|b-ic\|_2,\\
h_{21} &=& \|b+ic\|_2,\\
h_{22} &=& \frac{1}{2}\left(\left(1-\frac{1}{\epsilon}\right)\|a+\epsilon d\|_2 + \left(1+\frac{1}{\epsilon}\right)\|a-\epsilon d\|_2 \right).
\eeas
Note that $h$ is indeed Hermitian; $\|b-ic\|_2 = \|b+ic\|_2$, using the cyclicity of trace.
Now, by Lemma \ref{lem:basecase}, $\|g\|_q^q \le 2^{-(q/2-1)}
\|h\|_2^q = 2 \left(\frac{1}{2}\|h\|_2^2\right)^{q/2}$. An explicit
expansion gives
\be
\|h\|_2^2 = 2 \|b-ic\|_2^2 + \left(1+\frac{1}{\epsilon^2}\right) \left(\|a\|_2^2 + \epsilon^2 \|d\|_2^2 \right) + \left(1-\frac{1}{\epsilon^2}\right)\|a+\epsilon d\|_2\|a-\epsilon d\|_2,
\ee
implying that, in order to prove the proposition, we need
\be
\left(\frac{1}{2}+\frac{1}{2\epsilon^2}\right) \left(\|a\|_2^2 + \epsilon^2 \|d\|_2^2 \right) + \left(\frac{1}{2} - \frac{1}{2\epsilon^2}\right)\|a+\epsilon d\|_2\|a-\epsilon d\|_2 \le \|a\|_2^2 + \|d\|_2^2.
\ee
So, noting that $1/2-1/(2\epsilon^2)$ is negative, it suffices to show that
\be \|a\|_2^2 - \epsilon^2 \|d\|_2^2 \le \|a+\epsilon d\|_2\|a-\epsilon d\|_2.\ee
This last inequality can be proven using the matrix Cauchy-Schwarz inequality:
\bea \|a\|_2^2 - \epsilon^2 \|d\|_2^2 &=& \tr((a + \epsilon d)(a - \epsilon d))\\
&\le& \sqrt{\tr((a + \epsilon d)^2)}\sqrt{\tr((a - \epsilon d)^2)}\\
&=& \|a+\epsilon d\|_2\|a-\epsilon d\|_2,
\eea
so we are done.
\end{proof}

Theorem~\ref{thm:hypercontractivity} has the following easy
corollaries. The first says, informally, that low-degree quantum
boolean functions are smooth.

\begin{corollary}
\label{cor:hypercontractivity} Let $f$ be a Hermitian operator on $n$ qubits with
degree at most $d$. Then, for any $q \ge 2$, $\|f\|_q \le (q-1)^{d/2} \|f\|_2$. Also, for any $p \le 2$, $\|f\|_p \ge (p-1)^{d/2} \|f\|_2$.
\end{corollary}

\begin{proof}
The proof follows that of Corollary 1.3 in Lecture 16 of \cite{odonnellnotes} with no changes required. Explicitly,
\bea
\|f\|_q^2 &=& \left\| \sum_{k=0}^d f^{=k} \right\|_q^2 = \left\| T_{1/\sqrt{q-1}} \left( \sum_{k=0}^d (q-1)^{k/2} f^{=k} \right) \right\|_q^2\\
&\le& \left\| \sum_{k=0}^d (q-1)^{k/2} f^{=k} \right\|_2^2 = \sum_{k=0}^d (q-1)^k \sum_{\mathbf{s},|\mathbf{s}|=k} \fhats^2\\
&\le& (q-1)^d \sum_{\mathbf{s}} \fhats^2 = (q-1)^d \|f\|_2^2.
\eea
The second part is proved using the first part and H\"older's inequality, following immediately from
\be \|f\|_2^2 = \ip{f}{f} \le \|f\|_p \|f\|_q \le (q-1)^{d/2} \|f\|_p \|f\|_2, \ee
where $1/p + 1/q = 1$.
\end{proof}
The second corollary is a quantum counterpart of the fundamental Schwartz-Zippel Lemma \cite{schwartz:1980,zippel:1979}. This lemma states that any non-zero function $f:\{0,1\}^n \rightarrow \R$ of degree $d$ must take a non-zero value on at least a $2^{-d}$ fraction of the inputs. By analogy with the classical lemma, we conjecture that the constant in the exponent of this corollary can be improved.

\begin{corollary}
Let $f$ be a non-zero Hermitian operator on $n$ qubits with $m$
non-zero eigenvalues and degree $d$. Then $m
\ge 2^{n - (2\log e)d} \approx 2^{n - 2.89d}$.
\end{corollary}

\begin{proof}
Let $f_1$ denote the projector onto the subspace spanned by $f$'s eigenvectors that have non-zero eigenvalues. Then, for any $q \ge p \ge 1$,
\be \|f\|_p = \ip{f^p}{f_1}^{1/p} \le (\|f^p\|_{q/p} \|f_1\|_{q/(q-p)} )^{1/p} = \|f\|_q \left(\frac{m}{2^n}\right)^{1/p-1/q}, \ee
where we use H\"older's inequality and the fact that $f_1$'s non-zero eigenvalues are all 1. Thus, using Corollary \ref{cor:hypercontractivity}, for any $q \ge 2$ we have
\be \left(\frac{m}{2^n}\right)^{1/2-1/q} \ge \frac{\|f\|_2}{\|f\|_q} \ge (q-1)^{-d/2}, \ee
implying
\be m \ge \frac{2^n}{\left( (q-1)^{q/(q-2)} \right)^d}. \ee
Taking the limit of this expression as $q\rightarrow 2$ gives the desired result.
\end{proof}
The final corollary is a quantum generalisation of a lemma
attributed to Talagrand \cite{talagrand:1996a}, which bounds the
weight of a projector on the first level in terms of its 1-norm
(equivalently, its dimension). This lemma quantifies the intuition that low-rank projections (eg., pure states) on the Hilbert space of $n$ qubits must have a Fourier spectrum whose support includes high-weight Fourier coefficients.

\begin{corollary}
Let $P^2 = P$ be a projector. Then
\begin{equation}
\|P^{=1}\|_2^2 \le (q-1)\|P\|_1^{\frac{2}{p}},
\end{equation}
where $\frac{1}{p}+ \frac{1}{q} =1$.
\end{corollary}
\begin{proof}
We begin by writing
\begin{equation}
P = \sum_{\mathbf{s}}\hat{P}_{\mathbf{s}} \chis  = P^{=1} + h,
\end{equation}
where
\begin{equation}
P^{=1} = \sum_{\mathbf{s}\,|\, |\mathbf{s}| = 1}\hat{P}_{\mathbf{s}}
\chis.
\end{equation}
Note that $\|P^{=1}\|_2^2 = \ip{P^{=1}}{P}$. Applying H{\"o}lder's
inequality:
\begin{equation}
\|P^{=1}\|_2^2 = \ip{P^{=1}}{P} \le \|P\|_p\|P^{=1}\|_q,
\end{equation}
where $\frac{1}{p}+ \frac{1}{q} =1$. Next we use hypercontractivity
to show that $\|P^{=1}\|_q \le \sqrt{q-1}\|P^{=1}\|_2$, so that
\begin{equation}
\|P^{=1}\|_2^2 \le \sqrt{q-1}\|P^{=1}\|_2\|P\|_p =
\sqrt{q-1}\|P^{=1}\|_2\|P\|_1^{\frac{1}{p}}.
\end{equation}
Dividing both sides by $\|P^{=1}\|_2$ and squaring both sides gives
us
\begin{equation}
\|P^{=1}\|_2^2 \le (q-1)\|P\|_1^{\frac{2}{p}}.
\end{equation}
\end{proof}

%------------------------------------------------------------------------------
%------------------------------------------------------------------------------

\section{A quantum FKN theorem}

The Friedgut-Kalai-Naor (FKN) theorem \cite{friedgut:2002a} states that, if a boolean function has most of its Fourier weight on the first level or below, then it is close to being a dictator. It has proven useful in social choice theory \cite{kalai:2002a} and the study of hardness of approximation \cite{khot:2004a,dinur:2007a}. In this section, we state and prove two quantum variants of the FKN theorem. The first is a direct generalisation of the classical result, and uses the quantum hypercontractive inequality. The second is a different generalisation, to the $\infty$-norm.

%------------------------------------------------------------------------------

\subsection{Balancing quantum boolean functions}
\label{sec:balancing}

It will be convenient for the later results in this section to deal
only with quantum boolean functions which have no weight on level 0
(i.e.\ are traceless). We therefore describe a method to {\em
balance} quantum boolean functions. That is, from a quantum boolean
function $f$ we produce another quantum boolean function $g$ which
satisfies $\tr(g) = 0$ and has $\|f^{\le 1}\|_2^2 = \|g^{=1}\|_2^2$.

\begin{definition}
The \emph{spin flip} operation $S$ is defined as
\begin{equation}
\label{eq:spinflip}
S(M) = \sigma^2 M^* \sigma^2
\end{equation}
for any single qubit operator $M$.
\end{definition}
The spin flip operation is a superoperator but is not a CP map. Note
that $S(\sigma^j) = -\sigma^j$ for all $j \in \{1,2,3\}$.

\begin{definition}
The \emph{balancing operation} $\mathcal{B}$ is the superoperator
\begin{equation}
\begin{split}
\mathcal{B}(f) &= |0\rangle\langle0| \otimes f - |1\rangle\langle
1|\otimes S^{\otimes n}(f) \\
&= |0\rangle\langle0| \otimes f - |1\rangle\langle 1|\otimes
(\sigma^2\otimes \cdots \otimes\sigma^2)f^*(\sigma^2\otimes \cdots
\otimes\sigma^2),
\end{split}
\end{equation}
where we attach an ancilla qubit, denoted $A$.
\end{definition}

\begin{lemma}
Let $f$ be a quantum boolean function. Then $g = \mathcal{B}(f)$ is
a quantum boolean function.
\end{lemma}
\begin{proof}
Consider
\begin{equation}
\begin{split}
g^2 &= |0\rangle\langle0| \otimes f^2 + |1\rangle\langle 1|\otimes
S^{\otimes n}(f^2) \\
&= |0\rangle\langle0| \otimes \mathbb{I} + |1\rangle\langle
1|\otimes S^{\otimes n}(\mathbb{I}) = \mathbb{I}.
\end{split}
\end{equation}
\end{proof}

\begin{proposition}
Let $f$ be a quantum boolean function. Then $\tr(\mathcal{B}(f)) =
0$.
\end{proposition}
\begin{proof}
The proof is by direct calculation.
\begin{equation}
\tr(\mathcal{B}(f)) = \tr(f) - \tr(f^*) = 2\Im(\tr(f)) = 0,
\end{equation}
where the first equality follows from an application of the cyclic rule of trace.
\end{proof}

\begin{proposition}
Let $f$ be a quantum boolean function, and $g =\mathcal{B}(f)$. Then
$\|g^{\le 1}\|_2^2 = \|f^{\le 1}\|_2^2$.
\end{proposition}
\begin{proof}

We can calculate the first level of $\mathcal{B}(f)$ by tracing against
weight-$1$ operators on the system:
\begin{equation}
\begin{split}
\tr(\sigma_j^s\mathcal{B}(f)) &= \tr(\sigma_j^sf) - \tr(\sigma_j^s
(\sigma^2\otimes \cdots \otimes\sigma^2)(f^*)(\sigma^2\otimes \cdots
\otimes\sigma^2))
\\
&= 2\tr(\sigma_j^sf),
\end{split}
\end{equation}
recalling that the notation $\sigma^j_i$ is used for the operator which acts as $\sigma^j$ at the $i$'th
position, and trivially elsewhere. That is, $\ip{\sigma_j^s}{g} = \ip{\sigma_j^s}{f}$, so all the degree 1
terms in the Fourier expansion of $g$ (on the system) are identical to
those for $f$. The only non-zero weight-$1$ term of $g$ on the ancilla $A$ is
given by
\begin{equation}
\tr(\sigma^3_A g) = 2 \tr(f).
\end{equation}
\end{proof}
Our balancing operation reduces to the previously known classical balancing operation on
classical boolean functions \cite{friedgut:2002a}.

%------------------------------------------------------------------------------

\subsection{Exact quantum FKN}

To gain some intuition for the later results, we begin by sketching the (straightforward) proof of an exact variant of the FKN theorem.

\begin{proposition}
Let $f$ be a quantum boolean function on $n$ qubits. If $\sum_{|s|>1} \hat{f}_s^2 = 0$, then $f$ is either a dictator or constant.
\end{proposition}

\begin{proof}
By the results of Section \ref{sec:balancing}, we can assume that $f$ is balanced. Expand $f$ as
\be f = \sum_{i=1}^n U_i, \ee
where $U_i$ is a traceless operator that acts non-trivially on only the $i$'th qubit, with at most two distinct eigenvalues. Let the positive eigenvalue of the non-trivial component of $U_i$ be $\lambda_i$, and the corresponding eigenvector be $\ket{e_i}$. Then the tensor product $\bigotimes_i \ket{e_i}$ is an eigenvector of $f$ with eigenvalue $\lambda = \sum_i \lambda_i \le \|f\|_\infty$. Let $\mathbf{i:x}$ denote the string of length $n$ with value $x$ at position $i$, and 0 elsewhere. Because $U_i = \mathbb{I}_{[i-1]}\otimes M_i \otimes \mathbb{I}_{[i+1,\ldots, n]}$ and $M_i = \hat{f}_{\mathbf{i:1}}\sigma_i^1 + \hat{f}_{\mathbf{i:2}}\sigma^2_i + \hat{f}_{\mathbf{i:1}}\sigma^3_i = \left(\begin{smallmatrix} \hat{f}_{\mathbf{i:3}} & \hat{f}_{\mathbf{i:1}}-i\hat{f}_{\mathbf{i:2}} \\
\hat{f}_{\mathbf{i:1}}+i\hat{f}_{\mathbf{i:2}} & -\hat{f}_{\mathbf{i:3}}\end{smallmatrix}\right)$, we have that
\be \lambda_i = \sqrt{\hat{f}_{\mathbf{i:1}}^2+\hat{f}_{\mathbf{i:2}}^2+\hat{f}_{\mathbf{i:3}}^2}. \ee
Thus $\sum_i \lambda_i$ is strictly greater than 1 unless $f$ is a dictator.
\end{proof}

%------------------------------------------------------------------------------

\subsection{Quantum FKN in the 2-norm}

In this section, we will prove the following result.

\begin{theorem}
\label{thm:quantumfkn}
There is a constant $K$ such that, for every quantum boolean function $f$, if $\sum_{|s|>1} \hat{f}_s^2 < \epsilon$, $f$ is $K\epsilon$-close to being a dictator or constant.
\end{theorem}
The proof is essentially a quantisation of one of the two proofs of the classical theorem given in the original paper \cite{friedgut:2002a}; see also the exposition in Lecture 13 of \cite{odonnellnotes}. We will require the following lemma.

\begin{lemma}
\label{lem:deg2}
Let $q$ be a degree 2 Hermitian operator on $n$ qubits such that $\Pr_i[|\lambda_i(q)| > \delta] = p$, where the probability is taken with respect to the uniform distribution on the eigenvalues of $q$. Then $\|q\|_2^2 \le \frac{\delta^2(1-p)}{1-9\sqrt{p}}$.
\end{lemma}

\begin{proof}
Expand $q=r\oplus s$, where $r$ projects onto the eigenvectors of $q$ with eigenvalues at most $\delta$ in absolute value. Thus $\mathrm{rank}(r) \le (1-p)2^n$ and $\mathrm{rank}(s)=p2^n$. Then the lemma follows from
\beas
\|q\|_2^2 &=& \|r\|_2^2 + \|s\|_2^2 \le (1-p) \delta^2 + \sqrt{p} \|s\|_4^2\\
&\le& (1-p) \delta^2 + \sqrt{p} \|q\|_4^2 \le (1-p) \delta^2 + 9 \sqrt{p} \|q\|_2^2,
\eeas
where the final inequality is Corollary \ref{cor:hypercontractivity}.
\end{proof}

We now turn to the proof of Theorem \ref{thm:quantumfkn}.

\begin{proof}[Proof of Theorem \ref{thm:quantumfkn}]
By the results of Section \ref{sec:balancing}, we can assume that $f$ has no weight on level 0, i.e.\ is traceless. Given $f = \sum_{|s|=1} \hat{f}_s \chis + \sum_{|s|>1} \hat{f}_s \chis$, call the first sum $l$ and the second $h$ (``low'' and ``high''). As $f$ is quantum boolean, $(l+h)^2 = \1$. Thus $l^2 + hl + lh + h^2 = \1$. On the other hand, by explicit expansion of $l$, we have
\be l^2 = \sum_{\mathbf{s},|\mathbf{s}|=1} \fhats^2 \1 +
\mathop{\mathop{\sum_{\mathbf{s},\mathbf{t},}}_{|\mathbf{s}|=|\mathbf{t}|=1,}}_{|\mathbf{s}\cap
\mathbf{t}|=0} \fhats \fhatt \chis \chit  = (1-\epsilon)\1 + q, \ee
where we define a new operator $q$. Our goal will be to show that $\|q\|_2^2$ is small, i.e.\ at most $K \epsilon$ for some constant $K$. The theorem will follow: if
\bea
K \epsilon \ge \|q\|_2^2 &=& \mathop{\mathop{\sum_{\mathbf{s},\mathbf{t},}}_{|\mathbf{s}|=|\mathbf{t}|=1,}}_{|\mathbf{s}\cap \mathbf{t}|=0} \fhats^2 \fhatt^2
= \left(\sum_{\mathbf{s},|\mathbf{s}|=1} \fhats^2 \right)^2 - \mathop{
\mathop{\sum_{\mathbf{s},\mathbf{t},}}_{|\mathbf{s}|=|\mathbf{t}|=1,}}_{|\mathbf{s}\cap \mathbf{t}|=1} \fhats^2 \fhatt^2  \\
&=& (1 - \epsilon)^2 - \mathop{
\mathop{\sum_{\mathbf{s},\mathbf{t},}}_{|\mathbf{s}|=|\mathbf{t}|=1,}}_{|\mathbf{s}\cap \mathbf{t}|=1} \fhats^2 \fhatt^2,
\eea
then, for some $K'$,
\bea
1 - K' \epsilon &\le& \mathop{
\mathop{\sum_{\mathbf{s},\mathbf{t},}}_{|\mathbf{s}|=|\mathbf{t}|=1,}
}_{|\mathbf{s}\cap \mathbf{t}|=1} \fhats^2 \fhatt^2 = \sum_{\mathbf{s},|\mathbf{s}|=1} \fhats^2 \Big( \mathop{\sum_{\mathbf{t},|\mathbf{t}|=1,}}_{|\mathbf{s}\cap \mathbf{t}|=1} \fhatt^2 \Big)\\
&\le& \max_{\mathbf{s},|\mathbf{s}|=1} \mathop{\sum_{\mathbf{t},|\mathbf{t}|=1,}}_{|\mathbf{s}\cap \mathbf{t}|=1} \fhatt^2,
\eea
so there exists some index $i$ such that all but $K' \epsilon$ of the weight of $f$ is on terms that only depend on the $i$'th qubit. Setting all the other terms in the Fourier expansion of $f$ to zero and renormalising gives a quantum boolean function that is a dictator (on the $i$'th qubit) and is distance $K'' \epsilon$ from $f$, for some other constant $K''$.

It remains to show that $\|q\|_2^2$ is small. Expand $q$ as
\be q = \epsilon \1 - hl - lh - h^2 = \epsilon \1 - h(f-h) - (f-h)h - h^2 = \epsilon \1 - hf - fh + h^2 \ee
and consider the terms in this sum.  By the hypothesis of the theorem, $\|h\|_2^2 = \frac{1}{2^n} \sum_i \lambda_i(h)^2 = \epsilon$. We also have $\tr h=0$. By Chebyshev's inequality, this implies that, for any $K>0$, $\Pr_i[|\lambda_i(h)| > K \sqrt{\epsilon}] \le 1/K^2$ (taking the uniform distribution on the eigenvalues of $h$).

We also have $\sigma_i(hf) \le \sigma_i(h)$ (see \cite[Problem III.6.2]{bhatia:1997a} and note that $\|f\|_\infty=1$), and similarly for $\sigma_i(fh)$. As $h^2 \le \1$, for all $i$, $|\lambda_i| \le 1$ and so $\Pr_i[|\lambda_i(h^2)| > K \sqrt{\epsilon}] \le \Pr_i[|\lambda_i(h)| > K \sqrt{\epsilon}] \le 1/K^2$. This implies that (see \cite[Problem III.6.5]{bhatia:1997a})
\be \Pr_i [|\lambda_i(q)| > 3K \sqrt{\epsilon} + \epsilon] \le 3/K^2. \ee
Using Lemma \ref{lem:deg2}, taking $K$ to be a sufficiently small constant, the result follows.
\end{proof}

%------------------------------------------------------------------------------

\subsection{Quantum FKN in the $\infty$-norm}
In this subsection we'll prove a quantum generalisation of the FKN
theorem in terms of the \emph{supremum} or infinity norm $\|\cdot
\|_\infty$. Our proof doesn't make use of hypercontractivity: only
standard results from matrix analysis are employed.
\begin{theorem}\label{thm:inftyfkn}
Let $f$ be a quantum boolean function. If $\|f-g\|_\infty \le
\epsilon$, where $g$ is a Hermitian operator with $g = g^{= 1}$ and
$\epsilon < \frac12$, then $f$ is close to a dictator $h$, i.e.,
$\|f-h\|_\infty\le 2\epsilon$.
\end{theorem}

\begin{remark}
This result has no classical analogue as $\|f-g\|_\infty$ can never
be small if $f$ and $g$ are different.
\end{remark}

\begin{proof}
By the results of Section \ref{sec:balancing}, we can assume that
$f$ is traceless. Our proof works by using the infinity-norm
closeness of $f$ and $g$ in an application of Weyl's perturbation
theorem \cite{bhatia:1997a} to force the eigenvalues of $f$ to be
close to those of $g$:
\begin{equation}
|\lambda_j^\downarrow(f)-\lambda_j^\downarrow(g)| \le \epsilon,
\quad j= 1, 2, \ldots, 2^n,
\end{equation}
where $\lambda_j^{\downarrow}(M)$ denote the eigenvalues of $M$, in
descending order. Thus:
\begin{equation}
\lambda_j^\downarrow(f) = \lambda_j^\downarrow(g) + \epsilon(j),
\quad j= 1, 2, \ldots, 2^n.
\end{equation}
Since $g$ can be written as $g = \sum_{j\in[n]} g_j$, with
$\supp(g_j) = \{j\}$ and $\tr(g_j) = 0$, the eigenvalues of $g$ are given by
$\sum_{j=1}^n x_j \mu_j$, where $x_j \in \{-1, +1\}$ and $\mu_j =
\|g_j\|_\infty$. (This follows because $g_j$ is acts on the $j$th qubit as a 2-dimensional matrix.)

So we can label the eigenvalues of $g$ with $\mathbf{x} \in \{-1,
+1\}^n$ and we rewrite the perturbation condition as
\begin{equation}\label{eq:epseigeqs}
\lambda(\mathbf{x}) = \sum_{j=1}^n x_j \mu_j + \epsilon(\mathbf{x}),
\end{equation}
where $\lambda(\mathbf{x})$ is the eigenvalue of $f$ corresponding
to the eigenvalue $\sum_{j=1}^n x_j \mu_j$ of $g$ and
$\epsilon(\mathbf{x})$ is a correction. Our strategy is to now show
that $\epsilon(\mathbf{x})$ is a linear function.

To this end, we differentiate the $2^n$ equations
(\ref{eq:epseigeqs}) with respect to $x_j$:
\begin{equation}\label{eq:mujeps}
\mu_j = D_j\lambda - D_j\epsilon(\mathbf{x}),
\end{equation}
where, eg.,
\begin{equation}
D_j\lambda = \frac{\lambda(x_1, \ldots, x_{j}=+1, \ldots, x_n) -
\lambda(x_1, \ldots, x_{j}=-1, \ldots, x_n)}{2},
\end{equation}
i.e., we take the difference between $\lambda$ with the $j$th
variable assigned to $1$ and to $-1$. Note that $D_j\lambda :
\{-1,+1\}^{n-1} \rightarrow \mathbb{R}$.

Since $f$ is quantum boolean, by assumption, we have that
$|\lambda(\mathbf{x})| = 1$, $\forall \mathbf{x}$. So, because
$\mu_j$ is constant and $|D_j\epsilon(\mathbf{x})| \le \epsilon < 
1/2$, and (\ref{eq:mujeps}) is true for all $(x_1, \ldots, x_{j-1},
x_{j+1},\ldots x_n)$, we conclude that
$D_j\epsilon(\mathbf{x})$ is constant for all $(x_1, \ldots,
x_{j-1}$, $x_{j+1},\ldots x_n)$. (Otherwise we'd have a contradiction: as we run through all the assignments of $(x_1, \ldots, x_{j-1}$, $x_{j+1}, \ldots, x_n)$ the value of $D_j\lambda$ can, in principle, take both the values $0$ and $1$. However, owing to the constancy of $\mu_j$ and the fact that $|D_j\epsilon(\mathbf{x})| \le \epsilon < 
1/2$ only one of two possible values can be taken.) This is true for all $j \in [n]$. So
we learn that $\epsilon(\mathbf{x})$ is linear:
\begin{equation}
\epsilon(\mathbf{x}) = \sum_{j=1}^n x_j \epsilon_j.
\end{equation}
But the condition that $|\epsilon(\mathbf{x})| \le \epsilon$,
$\forall \mathbf{x}$, implies that
\begin{equation}
\sum_{j=1}^n |\epsilon_j|\le \epsilon.
\end{equation}
Summarising what we've learnt so far:
\begin{equation}
D_j\lambda = \mu_j + \epsilon_j.
\end{equation}
Since $D_j\lambda:\{-1,+1\}^{n-1}\rightarrow \{-1,0,1\}$ we must
have that $(\mu_j+\epsilon_j) \in \{-1,0,1\}$, $j\in [n]$. But this
actually means that there is exactly one $j$ for which $\mu_j =
1-\epsilon_j$ (or $-1-\epsilon_j$); the rest satisfy $|\mu_j| =
\epsilon_j$. The reason for this is as follows. Suppose there was
more than one such $j$. This would then lead to a contradiction as
one can always find an assignment of the variables $x_k$ so that
$|\lambda(\mathbf{x})| > 1$, contradicting the quantum booleanity of
$f$.

We now need to show that $f$ is in fact close to a dictator. We
define our dictator to be $h = \sgn(g)$. An application of the
triangle inequality gives us the result:
\begin{equation}
\begin{split}
\|f-h\|_\infty &\le \|f-g\|_\infty + \|g-h\|_\infty \\
&\le \epsilon + \sum_{j=1}^n |\epsilon_j| \le 2\epsilon
\end{split}
\end{equation}
\end{proof}

%------------------------------------------------------------------------------
%------------------------------------------------------------------------------

\section{Influence of quantum variables}
\label{sec:influence}
In this section we introduce the notion of \emph{influence} for
quantum boolean functions, and establish some basic properties of
the influence.

The classical definition of the influence of variable $j$ on a
boolean function $f$ is the probability that $f$'s value is undefined
if the value of $j$ is unknown, formally defined as
\be I_j(f) = \mathbb{P}_x [f(x) \neq f(x \oplus e_j)], \ee
where $x \oplus e_j$ flips the $j$th bit of $x$. In the quantum
case, we define the influence in terms of {\em derivative
operators}.

\begin{definition}
The $j$th derivative operator $d_j$ is the superoperator
\begin{equation}
d_j \equiv \frac{1}{2}(\mathcal{I}-S_j),
\end{equation}
where $\mathcal{I}$ is the identity superoperator and $S_j$ is the
spin flip operation (\ref{eq:spinflip}) on the $j$th qubit. Note
that
\begin{equation}
d_j(\chi_{\mathbf{s}}) =\begin{cases} \chi_{\mathbf{s}}, \quad s_j
\not= 0 \\
0, \quad s_j=0.\end{cases}
\end{equation}
The \emph{gradient} of $f$ is
\begin{equation}
\nabla f \equiv (d_1(f), d_2(f), \ldots, d_n(f)).
\end{equation}
The \emph{laplacian} of $f$ is
\begin{equation}
\nabla^2 f = \|\nabla f\|_2^2 = \sum_{j=1}^n \|d_j(f)\|_2^2.
\end{equation}
\end{definition}

The following lemma is immediate from the definition of the derivative operator.

\begin{lemma}
Let $f$ be an operator on $n$ qubits. Then the $d_j$ operator acts as follows.
\begin{equation}
d_j(f) = \sum_{\mathbf{s} | s_j \not= 0}
\hat{f}_{\mathbf{s}}\chi_{\mathbf{s}}.
\end{equation}
\end{lemma}

\begin{definition}
Let $f$ be a quantum boolean function. We define the
\emph{influence} of the $j$th qubit to be
\begin{equation}
I_j(f) \equiv \|d_j(f)\|_2^2,
\end{equation}
and the \emph{total influence} $I(f)$ to be
\begin{equation}
I(f) \equiv \sum_{j=1}^n I_j(f).
\end{equation}
\end{definition}
Note that this definition reduces to the classical definition when $f$ is
diagonal in the computational basis. Intuitively the quantum influence of the $j$th qubit measures the extent to which the value of a quantum boolean function is changed when in the input the state of $j$th qubit is inverted through the origin of the Bloch sphere.

\begin{proposition}
Let $f$ be a quantum boolean function. Then
\begin{equation}
I_j(f) = \sum_{\mathbf{s} | s_j \not= 0} \hat{f}_{\mathbf{s}}^2
\end{equation}
and
\begin{equation}
I(f) = \sum_{\mathbf{s}} |\mathbf{s}| \hat{f}_{\mathbf{s}}^2.
\end{equation}
\end{proposition}
\begin{proof}
Both results follow immediately from the definition of influence.
\end{proof}

The next result provides two other characterisations of the
derivative.
\begin{lemma}\label{lem:derivatives}
Let $f$ be a quantum boolean function. Then
\begin{equation}
\begin{split}
d_j(f) &= f - \tr_j(f)\otimes \frac{\1_j}{2}\\
&= f-\int dU \, U_j^\dag f U_j
\end{split}
\end{equation}
and
\begin{equation}
I_j(f) = \frac{1}{2}\int dU\, \|[U_j, f]\|_2^2,
\end{equation}
where the commutator $[U_j, f] \equiv U_j f - f U_j$, $dU$ is the Haar measure on $U(2)$ and $U_j \equiv
\mathbb{I}\otimes \cdots\otimes U \otimes \cdots \otimes \mathbb{I}$
with $\supp(U_j) = \{j\}$.
\end{lemma}

\begin{proof}
Both of the first two identities can be established by checking
$d_j$ on single qubit operators and extending by linearity.

The alternative characterisation of the influence can be proven as
follows.
\begin{equation}
\begin{split}
I_j(f) &= \left\|f-\int dU\, U_j^\dag f U_j\right\|_2^2 \\
&= \frac{1}{2^n}\int dUdV\,\tr\left( (f-U_j^\dag f U_j)(f-V_j^\dag f
V_j) \right)
\\ &= \frac{1}{2^n}\int dUdV\,\tr\left( \1 - U_j^\dag f U_jf - f V_j^\dag
f V_j + U_j^\dag f U_jV_j^\dag f V_j\right) \\
&= 1 - \frac{2}{2^n}\int dU\,\tr\left(U_j^\dag f U_jf\right) +
\frac{1}{2^n}\int
dUdV\,\tr\left( f U_jV_j^\dag f V_jU_j^\dag\right)\\
&= 1 - \frac{2}{2^n}\int dU\,\tr\left(U_j^\dag f U_jf\right) +
\frac{1}{2^n}\int
dU\,\tr\left( f U_j f U_j^\dag\right) \\
&= 1 - \frac{1}{2^n}\int dU\,\tr\left(U_j^\dag f U_jf\right)\\
&= \frac{1}{2^n}\int dU\,\tr\left(\1 - U_j^\dag f U_jf\right) \\
&= \frac{1}{2^{n+1}}\int dU\, \tr\left( [U_j,f][f,U_j^\dag] \right)
= \frac{1}{2}\int dU\, \|[U_j, f]\|_2^2.
\end{split}
\end{equation}
\end{proof}

Now we generalise the single-qubit influence to multiple qubits.

\begin{definition}
Let $f$ be a quantum boolean function. Then the influence of a set
$J\subset [n]$ on $f$, written $I_J(f)$, is the quantity
\begin{equation}
I_J(f) \equiv \|d_J(f)\|_2^2,
\end{equation}
where
\begin{equation}
d_J(f) \equiv f - \tr_J(f)\otimes \frac{\1_J}{2^{|J|}}.
\end{equation}
\end{definition}

The next result is a straightforward generalisation of
Lemma~\ref{lem:derivatives}.
\begin{corollary}
Let $f$ be a quantum boolean function, $J\subset [n]$, and $m =
|J|$. Then
\begin{equation}
d_J(f) = f-\int dU_1dU_2\cdots dU_m \, (U_1\otimes U_2\otimes\cdots
\otimes U_m)^\dag f (U_1\otimes U_2\otimes\cdots \otimes U_m)
\end{equation}
and
\begin{equation}
I_J(f) = \int dU_1dU_2\cdots dU_m \|[U_1\otimes U_2\otimes\cdots
\otimes U_m, f]\|_2^2,
\end{equation}
where $J = \bigcup_{j=1}^m\supp(U_j)$.
\end{corollary}

Unlike the definition of influence on a single qubit, the physical interpretation
of multiple qubit influence is less clear, and we leave it as an open question.

\begin{definition}
Let $f$ be a quantum boolean function. Then we define the
\emph{variance} of $f$ to be
\begin{equation}
\textrm{var}(f) = \frac{1}{2^n}\tr(f^2) -
\left(\frac{1}{2^n}\tr(f)\right)^2.
\end{equation}
\end{definition}

\begin{proposition}[Quantum Poincar\'e inequality for $n$ qubits]
Let $f$ be a quantum boolean function. Then
\begin{equation}
\text{\rm var}(f) \le \nabla^2(f) = I(f).
\end{equation}
\end{proposition}
\begin{proof}
The proof follows from writing left and right-hand sides in terms of
the Fourier expansion:
\begin{equation}
\text{var}(f) = \sum_{\mathbf{s}} \hat{f}_{\mathbf{s}}^2 -
f_{\mathbf{0}}^2
\end{equation}
and
\begin{equation}
I(f) = \sum_{\mathbf{s}} |\mathbf{s}|\hat{f}_{\mathbf{s}}^2.
\end{equation}
The inequality is now obvious.
\end{proof}

\begin{corollary}
\label{cor:poincare} Let $f$ be a quantum boolean function such that
$\tr(f) = 0$. Then there is a $j\in [n]$ such that $I_j(f)\ge 1/n$.
\end{corollary}

\section{Towards a quantum KKL theorem}
\label{sec:quantumkkl}

An influential paper of Kahn, Kalai and Linial \cite{kahn:1988a} proved the following result, known as the KKL theorem. For every balanced boolean function $f:\{0,1\}^n \rightarrow \{1,-1\}$, there exists a variable $x$ such that $I_j(f) = \Omega\left(\frac{\log n}{n}\right)$. By Corollary \ref{cor:poincare}, for every balanced quantum boolean function $f$ on $n$ qubits there is a qubit $j$ such that $I_j(f) \ge \frac{1}{n}$. It is thus natural to conjecture that a quantum analogue of the KKL theorem holds -- and also a quantum analogue of Friedgut's theorem \cite{friedgut:1998a}, which is based on similar ideas.

However, the immediate quantum generalisation of the classical proof
does not go through. Intuitively, the reason for this is as follows.
The classical proof shows that, if the influences are all small,
then their sum is large. This holds because, if the derivative
operator in a particular direction has low norm, then it has small
support, implying that it has some Fourier weight on a high level,
which must be included in derivatives in many different directions.
In the quantum case, this is not true: there exist quantum boolean
functions whose derivative is small in a particular direction, but
which are also low-degree.

On the other hand, it is immediate that the KKL theorem holds for
quantum boolean functions which can be diagonalised by local
unitaries, as these do not change the influence on each qubit. In
this section we describe three partial results aimed at generalising
the KKL theorem further. The first result is a simple quantisation
of one of the classical proofs. This serves to illustrate what goes
wrong when we generalise to the quantum world. Our next result shows
that the classical proof technique breaks down precisely for
anticommuting quantum boolean functions (qv.). Our final result is
then a stronger version of KKL for a class of anticommuting quantum
booolean functions.

\subsection{A quantum Talagrand's lemma for KKL}

The purpose of this section is to prove the following quantum
generalisation of a theorem of Talagrand \cite{talagrand:1994a}.

\begin{proposition}
\label{thm:talagrandkkl}
Let $f$ be a traceless Hermitian operator on $n$ qubits. Then
\be \|f\|_2^2 \le \sum_{i=1}^n \frac{10 \|d_i f\|_2^2}{(2/3) \log(\|d_i f\|_2/\|d_i f\|_1) + 1}. \ee
\end{proposition}
In the case of classical boolean functions, this result can be
applied to give an essentially immediate proof of the KKL theorem,
using the fact that the functions $\{d_i f\}$ take values in $\{-1,0,1\}$.
However, this does not extend to the quantum case, as the operators
$d_i f$ have no such constraint on their eigenvalues. The proof of
Proposition \ref{thm:talagrandkkl}, on the other hand, is essentially
an immediate generalisation of the classical proof in \cite{talagrand:1994a}
(alternatively, see the exposition in \cite{dewolf:2008a}).

\begin{proof}[Proof of Proposition \ref{thm:talagrandkkl}]
For any operator $g$, define
\be M^2(g) = \sum_{\mathbf{s}\neq 0} \frac{\ghats^2}{|\mathbf{s}|}. \ee
Then it is clear that
\be \|f\|_2^2 = \sum_{i=1}^n M^2(d_i (f)). \ee
Our strategy will be to find upper bounds on $M^2(g)$ for any traceless operator $g$. For some integer $m \ge 0$, expand
\bea
M^2(g) &=& \sum_{1 \le |\mathbf{s}| \le m} \frac{\ghats^2}{|\mathbf{s}|} + \sum_{|\mathbf{s}| > m} \frac{\ghats^2}{|\mathbf{s}|}\\
&\le& \sum_{k=1}^m \left( \frac{2^k}{k} \right) \|g^{=k}\|_{3/2}^2 + \frac{1}{m+1} \sum_{|\mathbf{s}| > m} \ghats^2\\
&\le& \|g\|_{3/2}^2 \sum_{k=1}^m \frac{2^k}{k} + \frac{1}{m+1} \|g\|_2^2, \eea
where we use quantum hypercontractivity (Corollary \ref{cor:hypercontractivity}) in the first inequality. In order to bound the first sum, we note that $\sum_{k=1}^m 2^k/k \le 4 \cdot 2^m/(m+1)$, which can be proved by induction, so
\be M^2(g) \le \frac{1}{m+1} \left( 4 \cdot 2^m \|g\|_{3/2}^2 + \|g\|_2^2 \right). \ee
Now pick $m$ to be the largest integer such that $2^m \|g\|_{3/2}^2 \le \|g\|_2^2$. Then $2^{m+1} \|g\|_{3/2}^2 \ge \|g\|_2^2$, and also $m+1 \ge 1$. Thus
\be m+1 \ge \frac{1}{2}\left(2 \log\left(\frac{\|g\|_2}{\|g\|_{3/2}}\right)+1\right), \ee
implying
\be M^2(g) \le \frac{10\,\|g\|_2^2}{2\log(\|g\|_2/\|g\|_{3/2})+1}. \ee
Noting that $\|g\|_2/\|g\|_{3/2} \ge (\|g\|_2/\|g\|_1)^{1/3}$, which can be proven
using Cauchy-Schwarz, and summing $M^2(d_i f)$ over $i$ completes the proof of the proposition.
\end{proof}

The worst case for this inequality is where the 2-norms and
$1$-norms of the operators $\{d_i f\}$ are the same. We therefore
address this case in the next section.

\subsection{A KKL theorem for anticommuting quantum boolean functions}
In this subsection we study the situation where the 2-norms and
$1$-norms of the operators $\{d_i f\}$ are the same. This situation is the ``worst-case scenario'' for a straightforward quantum generalisation of the classical proof. We show that if
this is the case then $f$ must be a sum of anticommuting quantum
boolean functions and so we identify this class as the ``most
quantum'' of quantum boolean functions. While this class might be expected to avoid a KKL-type theorem, we then study a subclass of such quantum boolean functions and the nevertheless provide a lower bound for the influence (which is better than a KKL-type bound).

\begin{definition}
Let $f$ be a quantum boolean function. Then the set $J\subset [n]$
of variables is said to have \emph{bad influence} on $f$ if
\begin{equation}
\left\|f - \tr_{J}(f)\otimes \frac{\mathbb{I}}{2^{|J|}}\right\|_2 =
\left\|f - \tr_{J}(f)\otimes \frac{\mathbb{I}}{2^{|J|}}\right\|_1.
\end{equation}
\end{definition}

\begin{lemma}\label{lem:12norm}
Let $M$ be an $n\times n$ Hermitian matrix. If $\|M\|_2 = \|M\|_1$
then the eigenvalues $\lambda_j$ of $M$ satisfy
\begin{equation}
\lambda_j\in\{-\alpha,+\alpha\}, \quad j \in [n],
\end{equation}
for some $\alpha \in \mathbb{R}$.
\end{lemma}
\begin{proof}
Diagonalising $M$ and writing out the equality $\|M\|_2^2 =
\|M\|_1^2$ gives us
\begin{equation}
\frac{1}{n}\sum_{j=1}^n |\lambda_j|^2 = \frac{1}{n^2}\left(\sum_{j
=1}^n |\lambda_j|^2 + 2 \sum_{j<k}^n |\lambda_j||\lambda_k|\right).
\end{equation}
Multiplying through by $n^2$ and rearranging gives us
\begin{equation}
(n-1)\sum_{j=1}^n |\lambda_j|^2 - 2 \sum_{j<k}^n
|\lambda_j||\lambda_k| = 0
\end{equation}
But this is the same as
\begin{equation}
\sum_{j<k} (|\lambda_j|-|\lambda_k|)^2 = 0,
\end{equation}
so that $|\lambda_j| = |\lambda_k| = \alpha$, $j < k$, for some
constant $\alpha$.
\end{proof}

The next lemma quantifies the structure of quantum boolean functions with bad influence. We will see that they are highly constrained: they must be a sum of anticommuting quantum boolean functions.

\begin{lemma}\label{lem:badinfstruct}
Let $f$ be a quantum boolean function. Then $J\subset [n]$ has bad
influence on $f$ if and only if
\begin{equation}
f = \sqrt{1-\alpha^2}f'\otimes \mathbb{I}_J + \alpha g
\end{equation}
where $f'$ and $g$ are quantum boolean functions, $\{f'\otimes
\mathbb{I}_J, g\} = 0$, and $\alpha^2 = I_J(f)$.
\end{lemma}
\begin{proof}
Write
\begin{equation}
f = f_{J^c}\otimes \mathbb{I}_J +
\sum_{\substack{\supp(\mathbf{s})\subset J \\
\mathbf{s}\not=\mathbf{0}}} f_{\mathbf{s}}\otimes \chi_\mathbf{s}.
\end{equation}
The derivative of $f$ with respect to $J$ is given by
\begin{equation}
d_J(f) = \sum_{\substack{\supp(\mathbf{s})\subset J \\
\mathbf{s}\not=\mathbf{0}}} f_{\mathbf{s}}\otimes \chi_\mathbf{s}.
\end{equation}
Since $J$ has bad influence we have, according to
Lemma~\ref{lem:12norm}, that
\begin{equation}
|d_J(f)| = \alpha \mathbb{I},
\end{equation}
so that, defining $g = \frac1\alpha d_J(f)$, we have $g^2 =
\mathbb{I}$. Thus we can write $f = f_{J^c}\otimes \mathbb{I}_J +
\alpha g$.

We now square $f$ to find
\begin{equation}
\mathbb{I} = f_{J^c}^2\otimes \mathbb{I}_J + \alpha^2\mathbb{I} +
\alpha\{f_{J^c}\otimes \mathbb{I}_J, g\}.
\end{equation}
This implies that $\{f_{J^c}\otimes \mathbb{I}_J, g\} = 0$. Thus,
defining
\begin{equation}
f' = \frac{1}{\sqrt{1-\alpha^2}}f_{J^c},
\end{equation}
we have $f'^2 = \mathbb{I}$ and the result follows.
\end{proof}

Let $f$ be quantum boolean with Fourier expansion $f = \sum_s \fhats
\chis$. If $[\chis,\chit]=0$ for all $s \neq t$ where $\fhats$ and $\fhatt$ are both non-zero, then we call $f$
{\em commuting}. Similarly, if $\{\chis,\chit\}=0$ for all $s \neq
t$ where $\fhats$ and $\fhatt$ are both non-zero, then we call $f$ {\em anticommuting}. It follows from the
classical KKL theorem that those commuting quantum boolean functions
that can be diagonalised by local unitaries
have a qubit with influence at least $\Omega\left(\frac{\log
n}{n}\right)$. In the remainder of this section, we will show that anticommuting
quantum boolean functions also have an influential qubit. Indeed,
the influence of this qubit must be very high.

\begin{proposition}
Let $f$ be an anticommuting quantum boolean function on $n$ qubits. Then there exists a $j$ such that $I_j(f) \ge \frac{1}{\sqrt{n}}$.
\end{proposition}

\begin{proof}
Our approach will be to show that $\sum_{j=1}^n I_j(f)^2 \ge 1$, whence the theorem follows trivially. Write $f = \sum_{i=1}^m w_i f_i$ for some $m$, where $w_i$ is real and $f_i$ is an arbitrary stabilizer operator. Let $S_j$ be the set of indices of the stabilizer operators that act non-trivially on qubit $j$, i.e.\ the set $\{i:d_j(f_i) \neq 0\}$. Then $I_j(f) = \sum_{i \in S_j} w_i^2$, so
\[ \sum_{j=1}^n I_j(f)^2 = \sum_{j=1}^n \left( \sum_{i \in S_j} w_i^2 \right)^2 = \sum_{j=1}^n \sum_{i,k \in S_j} w_i^2 w_k^2. \]
Rearrange the sum as follows.
\[ \sum_{j=1}^n I_j(f)^2 = \sum_{i,k=1}^m w_i^2 w_k^2 |\{j:i\in S_j, k \in S_j\}|. \]
For each pair of stabilizer operators $f_i$, $f_k$ to anticommute, they must both act non-trivially on the same qubit in at least one place. Thus
\[ \sum_{j=1}^n I_j(f)^2 \ge \sum_{i,k=1}^m w_i^2 w_k^2 = \left( \sum_{i=1}^m w_i^2\right)^2 = 1 \]
and the proof is complete.
\end{proof}

This result hints that quantum KKL may be true, as it holds for two ``extremal'' cases (classical boolean functions and anticommuting quantum boolean functions).

\section{Conclusions and conjectures}

We have introduced the concept of a quantum boolean function, and have quantised some results from the classical theory of boolean functions. However, there is still a hoard of interesting results which we have not yet been able to plunder. We list some specific conjectures in this vein below.

\begin{enumerate}
\item {\bf Quantum locality and dictator testing.} We have candidate quantum tests for the properties of locality and being a dictator (Conjectures \ref{con:locality} and \ref{con:dictator}), but have not been able to analyse their probability of success.

\item {\bf General hypercontractivity.} We have proven hypercontractivity of the noise superoperator (Theorem~\ref{thm:hypercontractivity}) only in the case where $1 \le p \le 2 \le q \le \infty$. We conjecture that, as with the classical case, this in fact holds for all $1 \le p \le q \le \infty$.

\item {\bf Every quantum boolean function has an influential variable.} The results
of Section \ref{sec:quantumkkl} prove a quantum generalisation of
the KKL theorem in some special cases. We conjecture that it holds
in general, but, as we argued, a proof of such a theorem appears to
require quite different techniques to the classical case.

\item {\bf Lower bounds on the degree of quantum boolean functions.} A classical result of Nisan and Szegedy \cite{nisan:1994a} states that the degree of any boolean function that depends on $n$ variables must be at least $\log n - O(\log \log n)$. We conjecture that the degree of any quantum boolean function that acts non-trivially on $n$ qubits is also $\Omega(\log n)$. The classical proof does not go through immediately: it relies on the fact that the influence of each variable of a degree $d$ boolean function is at least $1/2^d$, which is not the case for degree $d$ quantum boolean functions (for a counterexample, see (\ref{eq:smallcoeff})).
\end{enumerate}

It is interesting to note that the proofs in the classical theory of boolean functions which go through easily to the quantum case tend to be those based around techniques such as Fourier analysis, whereas proofs based on combinatorial and discrete techniques do not translate easily in general. There are many such combinatorial results which would be interesting to prove or disprove in the quantum regime.

\section*{Acknowledgements}

AM was supported by the EC-FP6-STREP network QICS. TJO was supported
by the University of London central research fund. We'd like to
thank Koenraad Audenaert, Jens Eisert, and Aram Harrow for helpful
conversations, and also Ronald de Wolf and a STOC'09 referee for helpful
comments on a previous version.

\bibliography{qinf}

\end{document}